\newif\ifacmversion
\newif\ifshortversion

\ifacmversion
\documentclass[submit,box]{acmconf}
\ConferenceName{18th ACM Conference on Computer and Communications Security}
\ConferenceShortName{ACM CCS 2011}
\shortversiontrue
\author{%
\Author{Yannick Chevalier}\\
\Address{IRIT -- Universit{\'e} de Toulouse, France}\\
\Email{ychevali@irit.fr}\\
}
\else
\documentclass{article}
\author{Yannick Chevalier}
\fi

\title{Finitary Deduction Systems}

\newtheorem{theorem}{Theorem}
\newtheorem{definition}{Definition}
\newtheorem{lemma}{Lemma}
\newtheorem{proposition}{Proposition}
\newtheorem{example}{Example}

\usepackage{amsmath,amssymb,amsfonts}
\usepackage{proof}
\usepackage{xspace}
\usepackage{comment}
\ifshortversion
\excludecomment{shortversion}
\else
\includecomment{shortversion}
\fi

\usepackage{tikz}
\usetikzlibrary{shapes.arrows,positioning,chains}

\newcounter{sdnum}
\newcounter{stepnum}[sdnum]
\newcounter{previousstepnum}[sdnum]

\makeatletter
\def\parseleftoutput#1#2{
\global\def\outputstate{t}\global\def\outleft{t}\global\def\sdnodetype{#2}}
\def\parserightoutput#1#2{
\global\def\outputstate{t}\global\def\outleft{f}\global\def\sdnodetype{#2}}
\def\parsenooutput#1{\global\def\outputstate{f}\global\def\sdnodetype{#1}}
\def\typeOf{\@ifnextchar w\parseleftoutput{\@ifnextchar e\parserightoutput\parsenooutput}}
\makeatother
\global\def\arrowlength{.5*\sdunitfactor}

\global\def\sdunitfactor{1}
\global\def\abscisse{3*\thesdnum*\sdunitfactor}
\global\def\ordonnee{\thestepnum*\sdunitfactor}

\global\def\nodename#1#2{sd#1-#2}
\global\def\currentnodename{\nodename{\varname}{\thestepnum}}

\global\def\sdnode{%
\draw (\abscisse,\ordonnee)
      node[circle,draw, 
           label=above right:{\nodevariable},
           label={[label distance=-3mm]below right:{\nodeequation}},
           name=\currentnodename]
           {\step};
\if t\inputstate 
\draw[red,<-]  (\currentnodename) --  \if t\inleft +(-\arrowlength,0)\else +(\arrowlength,0)\fi;
\fi
\if t\outputstate 
\draw[green,->] (\currentnodename) -- \if t\outleft +(-\arrowlength,0)\else +(\arrowlength,0)\fi ;
\fi
\if t\lastvisible \draw[->] (\nodename{\varname}{\thepreviousstepnum}) -- (\currentnodename);\fi
}

\global\def\symbolicderivation#1#2{
  \stepcounter{sdnum}
  \def\varname{#1}
  \draw (\abscisse,\ordonnee) node (sd\varname-\thestepnum) {SD \varname} ;
  \global\def\lastvisible{f}
  \foreach \step/\completetype/\param in {#2} {
    \stepcounter{stepnum}
    \global\def\inputstate{f}
    \global\def\outputstate{f}
    \global\def\nodevariable{\ensuremath{\varname_{\step}}}
    \global\def\nodeequation{}
    \expandafter\typeOf\completetype
    \if t\outputstate \fi 
    \if h\sdnodetype
       \global\def\lastvisible{f}
    \else
       \if r\sdnodetype
          \global\def\nodevariable{\ensuremath{{\varname}_{\param}}}
        \else\fi
        \if d\sdnodetype
           \global\def\nodeequation{\ensuremath{\nodevariable\unif\param}}
        \else\fi
        \if i\sdnodetype
           \global\def\inputstate{t}
           \if l\param
           \global\def\inleft{t}\else\global\def\inleft{f}\fi
        \fi
        \def\tracingall{0}
        \sdnode
        \global\def\lastvisible{t}
     \fi 
    \stepcounter{previousstepnum}
  } 
}

\def\connecttoright#1#2{\draw[color=blue,style=->] (#1) .. controls +(right:1cm) and +(left:1cm) .. (#2);}
\def\connecttoleft#1#2{\draw[color=blue,style=->] (#1) .. controls +(left:1cm) and +(right:1cm) .. (#2);}

\def\orderto#1#2{\draw[style=->] (#1) -- (#2);}
\global\def\sdunitfactor{0.8}
\usepackage{enumerate}
\usepackage[utf8]{inputenc}
\usepackage{paralist}
\usepackage{url}
\usepackage{algorithm,algorithmic}
\usepackage{subfig}
\usepackage{empheq}
\newcommand{\thmref}{}

\newcommand{\envthm}[4]{
\newenvironment{#1}{%
\begin{trivlist}\item {#3 #2.} #4}%
{\end{trivlist}}}

\newcommand{\tq}{\,\vert\,}
\envthm{rlemma}{Lemma~\ref{\thmref}, p.~\pageref{\thmref}}{\bf}{\it}
\envthm{rproposition}{Proposition~\ref{\thmref}, p.~\pageref{\thmref}}{\bf}{\it}
\envthm{rtheorem}{Theorem~\ref{\thmref}, p.~\pageref{\thmref}}{\bf}{\it}
\envthm{rcorollary}{Corollary~\ref{\thmref}, p.~\pageref{\thmref}}{\bf}{\it}
\envthm{rdefinition}{Definition~\ref{\thmref}, p.~\pageref{\thmref}}{\bf}{\it}

\makeatletter

\newenvironment{reftheorem}[1]{\def\thmref{#1}\begin{rtheorem}}{\end{rtheorem}}
\newenvironment{reflemma}[1]{\def\thmref{#1}\begin{rlemma}}{\end{rlemma}}
\newenvironment{refproposition}[1]{\def\thmref{#1}\begin{rproposition}}{\end{rproposition}}

\def\makeconstant#1#2{%
\expandafter\def\csname #1\endcsname{%
\ensuremath{\operatorname{#2}}}
}
\def\makefunction#1#2{%
  \expandafter\newcommand\csname #1\endcsname[2][]%
  {\ensuremath{\operatorname{#2}_{\rm ##1}(##2)}}
}

\def\make@math@constant#1{\makeconstant{#1}{#1}}

\def\make@math@function#1{\makefunction{#1}{#1}}

\def\make@math@cal#1{%
\expandafter\def\csname #1\endcsname{%
\ensuremath{\mathcal{#1}}\xspace}
}

\def\make@math@ded#1{%
\expandafter\def\csname D#1\endcsname{%
\ensuremath{\mathcal{D}_{\textrm{\footnotesize #1}}}}
}

\def\make@math@iterator#1#2{\@for\p@arg:=#1\do{%
\edef\arg{\expandafter\@firstofone\p@arg}
\csname make@math@#2\endcsname\arg\relax
}}

\def\MakeLabelledFunctions#1{\make@math@iterator{#1}{labelledfunction}}
\def\MakeFunctions#1{\make@math@iterator{#1}{function}}
\def\MakeConstants#1{\make@math@iterator{#1}{constant}}
\def\MakeCal#1{\make@math@iterator{#1}{cal}}
\def\MakeDed#1{\make@math@iterator{#1}{ded}}
\makeatother

\def\call#1{\ensuremath{\mathcal{#1}}}

\newenvironment{ilitz}{\begin{inparaenum}[\itshape a\upshape)]}{\end{inparaenum}}

\MakeCal{X,C,F,P,D,I,E,V,K}

\makeatletter
\let\Res\@undefined
\makeatother
\newcommand{\Cons}[2][]{\ensuremath{\operatorname{Const}_{\rm #1}(#2)}}

\newcommand{\mgu}[2][]{\ensuremath{\operatorname{mgu}_{\rm #1}(#2)}}

\newcommand{\Res}[2][]{\ensuremath{\operatorname{Res}_{\rm #1}(#2)}}

\newcommand{\Sub}[2][]{\ensuremath{\operatorname{Sub}_{\rm #1}(#2)}}
\newcommand{\Supp}[2][]{\ensuremath{\operatorname{Supp}_{\rm #1}(#2)}}

\newcommand{\Var}[2][]{\ensuremath{\operatorname{Var}_{\rm #1}(#2)}}

\newcommand{\atoms}[2][]{\ensuremath{\operatorname{atoms}_{\rm #1}(#2)}}

\newcommand{\pk}[2][]{\ensuremath{\operatorname{pk}_{\rm #1}(#2)}}
\newcommand{\sk}[2][]{\ensuremath{\operatorname{sk}_{\rm #1}(#2)}}
\newcommand{\enc}[2][]{\ensuremath{\operatorname{enc}_{\rm #1}(#2)}}
\newcommand{\dec}[2][]{\ensuremath{\operatorname{dec}_{\rm #1}(#2)}}

\MakeConstants{Sat,min}

\makefunction{fosymbol}{Symb}

\newcommand{\Nat}{\ensuremath{\text{\rm{}I\kern-0.20em{}N}}}
\newcommand{\Rat}{\ensuremath{\text{\sf{}l\rm{}\kern-0.42em{}Q}}}
\newcommand{\Real}{\ensuremath{\text{\rm{}I\kern-0.20em{}R}}}

\newcommand{\Constants}{\ensuremath{C}}
\newcommand{\Variables}{\ensuremath{\mathcal{X}}\xspace}
\newcommand{\atomes}{\ensuremath{\mathcal{A}}\xspace}
\newcommand{\sig}[1]{\ensuremath{\mathcal{T}(#1)}\xspace}
\newcommand{\gsig}[1]{\ensuremath{\mathcal{T}(#1)}\xspace}
\newcommand{\vsig}[1]{\sig{#1,\Variables}\xspace}
\newcommand{\set}[1]{\ensuremath{\lbrace #1 \rbrace}}
\newcommand{\condset}[2]{\set{#1 \,\vert{}\, #2}}

\newcommand{\interpret}[2][\call{I}]{\ensuremath{[\kern-0.15em[#2]\kern-0.15em]_{#1}}}
\newcommand{\unif}{\ensuremath{\stackrel{?}{=}}}
\newenvironment{decisionproblem}[1]{

\vspace*{-2ex}

\begin{tabbing}
  \underline{\textbf{#1}}\\
  \hspace*{2em}\= \textbf{Input:}~ \=}{
\end{tabbing}

\vspace*{-2em}

}

\newcommand{\entreeu}[1]{ \hbox{\vbox{\parbox[t]{0.8\linewidth}{#1}}}\\}
  
\newcommand{\sortie}[1]{
\> \textbf{Output:} \> \hbox{\vbox{\parbox[t]{0.8\linewidth}{#1}}}\\}

\newcommand{\minic}{\ensuremath{c_{\text{min}}}}
\newcommand{\Nonces}{\ensuremath{C_{\text{new}}}}
\newcommand{\intrus}[3]{\ensuremath{(#1,#2,#3)}}
\newcommand{\defsd}[1]{\ensuremath{(\call V #1,\call S #1,\call K #1,\Invar #1,\Outvar #1)}}

\newcommand{\Ind}{\mbox{\sc{Ind}}\xspace}
\newcommand{\Invar}{\mbox{\sc{In}}\xspace}
\newcommand{\Outvar}{\mbox{\sc{Out}}\xspace}

\newcommand{\trace}[2]{\ensuremath{\text{\rm Tr}_{#2}(#1)}}
\newcommand{\open}[2]{\ensuremath{\text{\rm open}_{#2}(#1)}}
 
\newcommand{\rhnorm}[2]{\ensuremath{{(#1)\!\!\downarrow_{#2}}}}

\newcommand{\norm}[1]{\rhnorm{#1}{}}
\newcommand{\normo}[1]{\rhnorm{#1}{\call{O}}}

\def\comp{\circ}
\newcommand{\DY}{\ensuremath{\mathcal{{DY}}}\xspace}

\newcommand\union\cup

\ifshortversion
\excludecomment{proof}

\else
\newcounter{claimc}
\newenvironment{claim}{\stepcounter{claimc}\begin{trivlist}\item \textbf{Claim~\arabic{claimc}.} \rm }{\end{trivlist}}
\newenvironment{proofclaim}{
  \begin{itemize}
  \item[]\textrm{Proof of the claim. }\it}
{\hfil$\Diamond$\end{itemize}}

\newenvironment{proof}{\setcounter{claimc}{0}\begin{trivlist}\item \textsc{Proof.} \sl}{\hfil$\Box$\end{trivlist}}

\fi

\makeconstant{minic}{c_{min}}
\makeconstant{speC}{C_{spe}}

\newcommand{\penc}[2]{\enc[p]{#1,#2}}
\newcommand{\pdec}[2]{\dec[p]{#1,#2}}

\newcommand{\paire}[2]{\ensuremath{\left\langle #1,#2\right\rangle}}

\newcommand{\sfree}{\ensuremath{\text{\rm sf}}}

\begin{document}

\maketitle

\begin{abstract}
  Cryptographic protocols are the cornerstone of security in
  distributed systems. The formal analysis of their properties is
  accordingly one of the focus points of the security community, and
  is usually split among two groups. In the first group, one focuses
  on trace-based security properties such as confidentiality and
  authentication, and provides decision procedures for the existence
  of attacks for an on-line attackers. In the second group, one focuses
  on equivalence properties such as privacy and guessing attacks, and
  provides decision procedures for the existence of attacks for an
  offline attacker. In all cases the attacker is modeled by a
  deduction system in which his possible actions are expressed.

  We present in this paper a notion of finitary deduction systems that
  aims at relating both approaches. We prove that for such deduction
  systems, deciding equivalence properties for on-line attackers can be
  reduced to deciding reachability properties in the same setting.
\end{abstract}

\section{Introduction}
\label{equiv:sec:intro}

\paragraph{Context.}
Security protocols, \textit{i.e.} protocols in which the messages are
cryptographically secured, are a cornerstone of security in
distributed applications. The need for optimizing resource utilization
and their distributed nature make their design error prone, and formal
methods have been applied successfully to detect errors in the
past~\cite{Lowe96,ArmandoCCCT08}.  But they are limited in
expressiveness since in most cases authors either were focused on the
resolution of reachability problems, or considered models in which the
attacker could not interfere with the on-going communications among
the honest agents. In contrast we consider in this paper the general
case of equivalence properties \textit{w.r.t.} an on-line attacker.

Formal models of cryptographic protocols usually present the reader
with a dichotomy between the honest agents---translated into a
constraint system~\cite{AmadioL00,MS01,RT01} or a
frame~\cite{abadi97calculus}---, and the attacker---modeled by a
deduction system expressing its possible actions.  In contrast we have
introduced in~\cite{DBLP:conf/frocos/ChevalierLR07} a notion of
\emph{symbolic derivation} that unifies the honest and dishonest agent
models: the actions of all agents are represented by a sequence of
deductions, nonce creation, and communication actions.  The notion of
equivalence considered in this paper is the one of symbolic
derivations representing honest agents.

\paragraph{Intuition.} First, a trivial remark: since one can
construct deduction systems for which reachability is decidable but
static equivalence is not, it is clear that generally speaking being
able to decide reachability does not imply being able to decide
symbolic equivalence. However, in most cases, one can model
reachability as the satisfiability of a constraint system, and
describe the decision procedure using constraint transformation rules.
A \emph{solved form} is defined as a constraint system in which the
attacker just has to instantiate variables by any term he can
construct.  In practice, the proof of completeness of the procedure
consists in assuming the existence of a sequence of deduction steps
that satisfies the constraint system, and in proving that as long as
one such sequence exists, either the constraint system is in solved
form or there exists a transformation rule applicable on the
constraint system. Then, an argument is given to prove that there is
no infinite sequence of transformations. Using K{\"o}nig's lemma, the
finiteness (also to be proved) of the number of possible successors of
each constraint system implies termination of the procedure. 

Our motivation was that such procedures actually do much more than
simply deciding reachability, as they end with a set of constraint
systems in solved form that, as long as the completeness proof is
along the lines given above, cover all possible attacks. Formalizing
this argument is however not trivial, since
\begin{itemize}
\item not all instances of the variables occurring in a constraint
  system in solved form correspond to attacks; and
\item when testing the equivalence of two protocols, we have to take
  into account the equality tests the attacker can perform to analyze
  the responses of the honest agents.
\end{itemize}
We have bypassed the first difficulty by imposing that the attacker
instantiates the first-order variables in a constraint system in
solved form with constants, and proved that replacing these constants
by any possible construction yields another attacks. This replacement
is formalized by on ordering on the attacks, the attacks corresponding
to solved forms being the minimal ones. Finitary deduction systems are
those for which the set of minimal attacks is always finite. The
second difficulty is solved by first proving that it suffices to
consider an attacker that performs at most one test, and then proving
that this test can be guessed \emph{before} the computation of solved
forms. Finally and implementationwise, we consider \emph{effective}
finitary deduction system, for which we assume that this finite set is
computable.

\paragraph{Applications.} The symbolic equivalence notion we consider
in this paper has three straightforward applications, related
respectively to on-line guessing attacks, to proving cryptographic
properties in a symbolic setting, and to privacy. We have proved, in
collaboration with
M.~Rusinowitch~\cite{DBLP:journals/ipl/ChevalierR10} that every
protocol narration (for any deduction system) can be compiled into an
active frame, which is a simplified form of symbolic derivations with
a total ordering on states and no intermediate computations between
communications.

\emph{Guessing attacks.} Introduced by Schneier~\cite{Schneier94}
under the name of \emph{dictionary attacks}, they consist in guessing a
secret piece of data, and then being able to check whether the guess
is correct. They can be offline, in which case the attacker observes
interactions between honest participants and has to decide whether the
guessed piece of data has been employed, or on-line, in which case the
intruder can interact with the honest participants.

Guessing attacks have been formalized thanks to the concept of
indistinguishability (see \textit{e.g.}~\cite{AbadiBW06}).  We can say
now that a protocol is vulnerable to undetectable on-line guessing
attacks whenever \textit{(i)} the honest agents cannot distinguish
between a session with the right piece of data and one involving a
wrong guess, whereas \textit{(ii)} the intruder can distinguish the
two executions. We model the first point by stating that the tests
performed by the honest agents succeed in both cases, and the second
point by saying that the two executions are not equivalent.

\emph{Cryptographic properties.} A line of works initiated
by~\cite{abadirogaway} showed that computational proofs of
indistinguishability ensuring the security of a protocol can be
derived, under some natural hypothesis on cryptographic primitives,
from symbolic equivalence proofs. This has opened the path to the
automation of computational proofs.  It was shown by
\cite{Comon-LundhC08} that \emph{in presence of an active attacker}
observational equivalence of the symbolic processes can be transferred
to the computational level.

\emph{Privacy.}  Symbolic equivalence is a crucial notion for
specifying security properties such as anonymity or secrecy of a
ballot in vote protocols~\cite{kremer-vote}.  More generally, the
analysis of privacy, \textit{e.g.}  client's identity in an
anonymization protocol such as IDEMIX~\cite{FormalIDMX,cms10}, in
communication protocols is inherently an equivalence problem. One has
to prove that a protocol preserves the \emph{strong secrecy} of an
attribute, \textit{i.e.}  that an observer cannot distinguish the
execution of a protocol transmitting this attribute's value, be it a
vote or her identity, from one in which a random piece of data is
exchanged.

\paragraph{Related works.}
We believe that Mathieu Baudet's modeling of attacks by instantiation
of \emph{second-order} variables~\cite{Baudet05} is the real
breakthrough that enabled the formal analysis of the equivalence
problem in the on-line attacker setting. Indeed, it was the first-time
that the actions of the attacker were represented explicitly in
solutions, instead of just keeping track (with a substitution on the
first-order variables of the constraint system) of their interaction
with the honest participants.

In collaboration with
M.~Rusinowitch~\cite{DBLP:journals/ipl/ChevalierR10} we have given
another proof of Baudet's result in the setting of symbolic
derivations. We believe that this setting is more complex but
introduces a langage fit to prove decidability and complexity results.
Also it possesses a symmetry between honest participants and the
attacker that permits to greatly simplify otherwise redundant proofs.
We consider in this paper a setting in which the actions of the honest
agents are represented by one \emph{Honest symbolic derivation} (HSD)
and those of a \emph{unique} intruder by one \emph{Attacker Symbolic
  Derivation} (ASD).  Symbolic derivations can be seen as standing
between symbolic traces~\cite{Baudet05} and the simple cryptographic
processes of~\cite{CortierD09}: the sequence of messages is not
totally ordered as it is the case in~\cite{Baudet05}, but there is no
branching but for termination on error nor any recursive process.

Few decidability results are available. In the article~\cite{huttel}
H{\"u}ttel proves decidability for a fragment of the spi-calculus without
recursion for framed bisimilarity. Since, the only original
decidability result on the equivalence of symbolic traces\footnote{a
  restriction of symbolic equivalence in which the actions of all the
  honest agents are totally ordered.}  we are aware of is for the
class of subterm deduction systems and was given by
M.~Baudet~\cite{Baudet05,Baudet-these}. We have recently given another
proof of this result~\cite{jar2010}, on which this paper elaborates.
Implementation-wise, an efficient procedure is presented
in~\cite{ijcar2010} in which one considers only the Dolev-Yao
deduction system. In spite of the relevance of this problem, we are
not aware of any extension of Baudet's decidability results to other
classes of deduction systems.

In~\cite{DBLP:conf/csfw/TiuD10} the authors consider, as
H{\"u}ttel~\cite{huttel}, the same problem in the simpler case of the
standard Dolev-Yao syntactic deduction system (with no equational
theory). They employ the notion of solved form as introduced
in~\cite{AmadioL00}, and more specifically that solved forms cover all
possible attacks. The existence of such a finite set of solved forms
corresponds exactly to our notion of finitary deduction system.

However, we note that their setting enforces a strict separation
between the values of the first order variables and the observer
process. This has in our opinion two negative side-effects. First, it
is well-known that not all instances of the first-order substitutions
constructed are instances of attacks. Second, given that the authors
of~\cite{DBLP:conf/csfw/TiuD10} only keep track of the constraints
that remain to be solved, the attacks themselves are not represented
explicitly in the solution. Hence it is not possible to reason on all
first-order instances of a solved form (since they are not all
attacks) nor on the observer processes (since only their interaction
with the processes under scrutiny is recorded). This is the reason why
we believe that the symbolic derivation setting adopted in this paper,
while more cumbersome at first, is better suited to reason on sets of
solutions, and therefore on process equivalence.

Many works have been dedicated to proving correctness properties of
cryptographic protocols using equivalences on process calculi.  In
particular \emph{framed bisimilarity} has been introduced by Abadi and
Gordon~\cite{abadi97calculus} for this purpose, for the spi-calculus.
Another approach that circumvents the context quantification problem
is presented in ~\cite{boreale99proof} where labeled transition
systems are constrained by the knowledge the environment has of names
and keys. This approach allows for more direct proofs of equivalence.

In \cite{CortierD09} the authors show how to apply the result by
Baudet on S-equivalence to derive a decision procedure for symbolic
equivalence for subterm convergent theories for simple processes.
Since \cite{CortierD09} relies on the proof of Baudet's result, that
is long and difficult \cite{Baudet-these}, we believe that providing a
simple criterion will be useful to derive other decidability results
in process algebras.
  
To the best of our knowledge, the only tool (besides~\cite{ijcar2010})
capable of verifying equivalence-based secrecy is the resolution-based
algorithm of ProVerif~\cite{BlanchetOakland04} that has been extended
for handling equivalences of processes that differ only in the choice
of some terms in the context of the applied $\pi$-calculus
\cite{BlanchetAbadiFournetLICS05}. This allows to add some equational
theories for modeling properties of the underlying cryptographic
primitives.

\paragraph{Example finitary deduction systems.} We remark that the
standard Dolev-Yao deduction system~\cite{dolev83ieee} is finitary,
since for every attack one can guess a subsequence of deduction steps
which is itself an attack~\cite{ChevalierLR-LPAR07}. In this regard,
this work extends~\cite{DBLP:conf/csfw/TiuD10} to other deduction
systems such as subterm deduction systems (the proof that from every
attack one can guess a sequence of deductions bounded by the size of
the input protocol is given \textit{e.g.} in~\cite{Kourjieh-these}).
We leave to future work the extension to contracting saturated
deduction systems, also defined in~\cite{Kourjieh-these}.

\paragraph{Organization of this paper.}
We reuse in this paper the notions and notations for terms,
equational theories, deduction systems, and symbolic derivations
introduced in earlier papers (sections 2--3). We give in Section~4 a
few properties of symbolic derivations, and define finitary deduction
systems accordingly.  We present in Section~5 a sketch of the proof
the symbolic equivalence is decidable for finitary deduction systems,
and conclude in Section~6. 
\ifshortversion
A version of this article with the proofs
of all statements has been submitted to arxiv.
\else
This document is the version of an article submitted to ACM CCS 2011
with the addition of the proofs of all statements.
\fi

\section{Formal setting}

\subsection{Term algebra}

We consider a countable set of free constants \Constants{}, a
countable set of variables \Variables, and a signature \call{F}
(\textit{i.e.} a set of function symbols with arities). We denote by
\gsig{\call{F}} (resp.  \vsig{\call{F}}) the set of terms over
$\call{F}\cup{}\Constants{}$ (resp.
$\call{F}\cup{}\Constants{}\cup\Variables$).  The former is called the
set of ground terms over \call{F}, while the latter is simply called
the set of terms over \call{F}.  Variables are denoted by $x$, $y$,
terms are denoted by $s,t,u,v,\ldots$, and decorations thereof,
respectively.

A \emph{constant} is either a \emph{free} constant in \Constants{} or a
function symbol of arity $0$. Given a term $t$ we denote by \Var{t}
the set of variables occurring in $t$ and by \Cons{t} the set of
constants occurring in $t$.  We denote by \atoms{t} the set $\Var{t}
\cup \Cons{t}$. We denote by \atomes ~the set of all constants and
variables.  A substitution $\sigma$ is an idempotent mapping from
\Variables{} to \vsig{\call{F}} such that $\Supp{\sigma}= \{x|
\sigma(x)\not=x\}$, the \emph{support} of $\sigma$, is a finite set.
The application of a substitution $\sigma$ to a term $t$ is denoted
$t\sigma$ and is equal to the term $t$ where all variables $x$ have
been replaced by the term $x\sigma$. A substitution $\sigma $ is
\emph{ground} w.r.t.  $\call{F}$ if the image of $\Supp{\sigma }$ is
included in $\gsig{\call{F}}$.  

The set of the \emph{subterms} of a term $t$, denoted $\Sub{t}$, is
defined inductively as follows. If $t$ is a constant or a variable
then $\Sub{t}=\set{t}$. Otherwise, $t$ must be of the form
$f(t_1,\ldots,t_n)$, and we define
$\Sub{t}=\set{t}\cup\bigcup_{i=1}^n\Sub{t_i}$. The \emph{positions} in
a term $t$ are defined recursively as usual (\textit{i.e.} as
sequences of integers), $\epsilon$ being the empty sequence.  We
denote by $t_{|p}$ the subterm of $t$ at position $p$.  We denote by
$t[p \leftarrow s]$ the term obtained by replacing in $t$ the
syntactic subterm $t_{|p}$ by $s$.

\subsection{Equational theories and Unification}

We consider in this paper an equational theory \call{E} that defines a
congruence on the terms in \vsig{\call{F}}. We assume it is
consistent, \textit{i.e.} that it has a model with more than one
element. \emph{Ordered rewriting}~\cite{dershowitz90rewrite} then
permits us to employ the unfailing completion procedure
of~\cite{hsiangr87} to produce a (possibly infinite) set of equations
for which ordered rewriting is convergent \emph{on ground terms}, its
\emph{$o$-completion}. In turn, this convergence permits us to
constructively choose one element in the congruence class of each
ground term $t$, called its \emph{normal form}, and denoted \norm{t}.
We use in this paper the fact that since ordered rewriting is a
relation on ground terms, if a term $t$ is ground then the term
\norm{t} is also a ground term.

This construction relies on the assumption that the ground terms are
totally ordered by a simplification ordering, and that the minimum for
this ordering is a free constant \minic.

\subsubsection{Unification and equational theory type}

Our result on deduction systems may seem vacuous as the
definitions---based on an ordering on the ``attacks'' on a
protocol---are not constructive. They however follow a classical line
of definitions in the context of unification and equational theories.
We present in this subsection these classical notions (and refer the
reader \textit{e.g.} to~\cite{unification} for a more complete
overview) in order to hilight the similitudes between our definitions
and the classical ones for unification.

\begin{definition}{(\call{E}-unifiers)\label{def:fo:E:unification}}
  Let \call{E} be an equational theory. We say that two terms $t$ and
  $s$ are \emph{\call{E}-equal}, and denote $s=_{\call{E}}t$, if
  $\call{E}\models_= t=s$. We say that a substitution $\sigma$ is a
  \emph{\call{E}-unifier} of $s$ and $t$ if $\call{E}\models_=
  t\sigma=s\sigma$.
\end{definition}

We say that two terms that have a \call{E}-unifier are
\emph{\call{E}-unifiable}.

We denote $\Sigma_{\call{E}}(t,t')$ the set of all unifiers of $t$ and
$t'$. This set is not empty if, and only if, $t$ and $t'$ are
unifiable.  We extend the notion of unifier to conjunctions of
equations as follows.

\begin{definition}{(Unification systems)\label{def:protmod:unification}}
  Let $\call{E}$ be an equational theory.  An
  $\call{E}$-\emph{Uni\-fi\-ca\-tion system} $S$ is a finite set of
  equations denoted by $\set{u_i \unif{}
  v_i}_{i\in\set{1,\ldots,n}}$ with terms $u_i,v_i\in\vsig{\call{F}}$.  It is
  satisfied by a substitution $\sigma$, and we note $\sigma\models{}_\call{E}
  S$, if for all $i\in\set{1,\ldots,n}$ $u_i\sigma =_\call{E}
  v_i\sigma$.
\end{definition}

One defines an \emph{instantiation ordering} on unifiers by setting
$\sigma \le_i \tau$ whenever there exists a substitution $\theta$ such
that $\sigma\theta=_{\call{E}} \tau$.  Equational theories are
classified~\cite{Schmidt-Schauss86} \textit{w.r.t.} the possible
cardinalities of \emph{complete sets of unifiers}.
\begin{definition}{(Complete set of
      unifiers)\label{def:fo:complete:unifiers}}
    Let \call{E} be an equational theory and $t,t'$ be two terms. We
    say that a subset $S$ of $\Sigma_{\call{E}}(t,t')$ is a
    \emph{complete set of unifiers} of $t$ and $t'$ if, for every
    substitution $\sigma\in\Sigma_{\call{E}}(t,t')$ there exists a
    substitution $\tau\in S$ and a substitution $\theta$ such that
    $\tau\theta=_{\call{E}}\sigma$.
\end{definition}
Or, using the instantiation ordering terminology, a complete set of
unifiers is a set of minimal unifiers for the instantiation ordering
such that every unifier is an instance of a unifier in this set.
Finally, we define a set of \emph{most general unifiers} to be a
minimal set, for standard set inclusion, among the complete sets of
unifiers. The rationale for this definition is that modulo an
equational theory, two substitutions may be non-trivial instances one
of the other. In this case one of the two is redundant and can be
removed, hence the following definition.

\begin{definition}{(Most general \call{E}-unifiers)\label{def:fo:equational:mgu}}
  Let \call{E} be an equational theory. We call a set of most general
  \call{E}-unifiers of $t$ and $t'$, and denote \mgu[\call{E}]{t,t'},
  a minimal (for set inclusion) complete set of unifiers of two terms
  $t$ and $t'$.
\end{definition}
In the rest of this paper, and as long as there is no ambiguity, we
simply refer to such sets as sets of most general unifiers, or sets of
mgu.  Also, the notion of mgu is extended as usual to unification
systems. One proves the next lemma by constructing explicitly an
injection from each complete set of unifiers to the other.

\begin{lemma}{\label{lem:fo:cardinality:E-unifiers}}
  Let \call{E} be an equational theory, $t,t'$ be two terms, and
  $S,S'$ be two sets of most general unifiers of $t$ and $t'$.  Then
  $S$ and $S'$ have the same cardinality.
\end{lemma}

The finiteness or even the existence of a minimal complete set of
unifiers of two terms unifiable modulo \call{E} is not guaranteed. We
say that an equational theory is \emph{finitary} whenever, for every
two unifiable terms $t,t'$, \mgu[\call{E}]{t,t'} is a finite set.

One important property of unification systems that we shall use in the
rest of this paper is the following replacement property.

\begin{lemma}{\label{lem:unif:replacement}}
  For any equational theory \call{E}, if a \call{E}-unification system
  \call{S} is satisfied by a substitution $\sigma$, and $c$ is any
  free constant in \Constants{} away from $\call{S}$, then for any
  term $t$, $\sigma\delta_{c,t}$ is also a solution of \call{S}.
\end{lemma}

\paragraph{Variables and constants.} 
Using Lemma~\ref{lem:unif:replacement} we can clarify the difference
and similitudes between variables and free constants. First, a formal
point: since free constants do not occur in the equations of the
equational theory they are not among the constants obtained by
skolemization. Second, we agree that in the resolution
procedure~\cite{Robinson65}, variables have a special role whereas by
Herbrand's theorem we know that it suffices to consider models of a
set of clauses with at most one free constant. In spite of this we
almost use variables and free constants (as in
Lemma~\ref{lem:unif:replacement}) interchangeably.

The rationale is that ordered completion yields a rewriting relation
which is convergent on \emph{ground} terms, and thus cannot be
employed to normalize terms that contain variables.
Lemma~\ref{lem:unif:replacement} is thus fundamental since it implies
that some of the free constants that may appear in an unifier can be
replaced, the main difference with variables being that if, for a
simplification ordering $<$, we have $t<t'$, then for every
substitution $\sigma$ we also have $t\sigma<t'\sigma$, whereas it is
not the case that for every replacement $\delta_{c,s}$ we also have
$t\delta_{c,s}<t'\delta_{c,s}$.

\subsection{Deduction systems}
Our protocol analysis is based on the assumption that all the agents
operate on messages \textit{via} a message manipulation library.  We
consider a signature \call{F} containing the function symbols employed
to denote the messages, with a special subset of symbols $\call{F}_p$
denoting the functions of the library which can be employed by all
participants.

\begin{definition}{(Deduction
    systems)\label{def:protmod:deduction:system}}
  A \emph{deduction system} is defined by a triple
  $(\call{E},\call{F},\call{F}_p)$ where \call{E} is an equational
  presentation on a signature \call{F} and $\call{F}_p$ a subset of
  \emph{public} constructors in \call{F}.
\end{definition}

\begin{example}
  For instance the following deduction system models public key
  cryptography:
  $$
  \begin{array}{l}
    (  \set{\pdec{\penc{x}{y}}{y^{-1}} = x },\\
    \set{ \pdec{\_}{\_}, \penc{\_}{\_}, {\_^{-1}} }, \\
    \set{\pdec{\_}{\_},\penc{\_}{\_} } ) 
  \end{array}
  $$
  The equational theory is reduced here to a single equation that
  expresses that one can decrypt a cipher text when the inverse key is
  available.
\end{example}

\section{Symbolic derivations}
\label{sec:protmod:sd}

We present in this section our model for agents.

\subsection{Definitions}
\label{subsec:symb:def}

\paragraph{Symbolic derivations.}
Given a deduction system \intrus{\call{F}}{\call{P}}{\call{E}}, a role applies
public symbols in \call{P}{} to construct a response from its initial knowledge
and from messages received so far. Additionally, it may test equalities between
messages to check the well-formedness of a message.  Hence the activity of a
role can be expressed by a fixed symbolic derivation:

\begin{definition}{\label{def:symbolic:derivation}(Symbolic Derivations)}
  A symbolic derivation for a deduction system
  \intrus{\call{F}}{\call{P}}{\call{E}} is a tuple \defsd{} where \call{V}{}
  is a mapping from a finite ordered set $(\Ind,<)$ to a set of
  variables \Var{\call{V}}, \call{K}{} is a set of ground terms (the initial
  knowledge) \Invar{} is a subset of $\Ind$, \Outvar{} is a multiset
  of elements of $\Ind$ and {\call{S}} is a unification system.

  The set $\Ind$ represents internal states of the symbolic
  derivation.  We impose that any $i \in \Ind$ is exactly one of the
  following kind:
  \begin{description}
  \item[Deduction state:] There exists a public symbol $f\in{}\call{P}
    $ of arity $n$ such that
    $\call{V}{}(i)\unif{}f(\call{V}(\alpha_1),\ldots,\call{V}(\alpha_n))\in\call{S}$
    with $\alpha_j<i $ for $j\in\set{1,\ldots,n}$ .
  \item[Re-use state:] if there exists $j<i$ with
    $\call{V}(j)=\call{V}(i)$;
  \item[Memory state:] if there exists $t$ in \call{K} and an equation
    $\call{V}{}(i)\unif{}t$ in \call{S};
  \item[Reception state:] if $i \in \Invar{}$;
  \end{description}
  Additionally, a state $i$ is also an \textbf{emission state} if
  $i\in\Outvar$. 

  The unification system \call{S} contains no equation but those
  described above and equations $\call{V}(i)\unif \call{V}(j)$, and
  the mapping \call{V} must be injective on non-re-use states.
  
  A symbolic derivation is \emph{closed} if it has no reception
  state.  A substitution $\sigma$ \emph{satisfies} a closed
  symbolic derivation if $\sigma \models_\call{E} \call{S}$.
\end{definition}

We believe that using symbolic derivations instead of more standard
constraint systems permits one to simplify the proofs by having a more
homogeneous framework. There is however one drawback to their usage.
While most of the time it is convenient to have an identification
between the order of deduction of messages and their send/receive
order, building in this identification too strictly would prevent us
from expressing simple problems. Re-use states are employed to reorder
the deduced messages to fit an order of sending messages which can be
different. For example consider an intruder that knows (after
reception) two messages $a$ and $b$ received in that order, and that
he has to send first $b$, then $a$.  Since the states in a symbolic
derivation have to be ordered, we have to use at least one re-use
state (for $a$) to be able to consider a sending of $a$ \emph{after}
the sending of $b$. We note that re-use states that are not employed
in a connection can be safely eliminated without changing the
deductions, the definition of the knowledge nor the tests in the
unification system.

With respect to earlier definitions, we have chosen to consider
injective variable-state mapping functions. The rationale for this
choice is essentially aesthetic, as using this more strict definition
implies that every equality test performed by the attacker is an
equality $\call{V}(i)\unif\call{V}(j)$ in the unification system. Not
having this restriction would require the introduction of
\begin{ilitz}
\item an equivalence class on ASDs to model the fact that two ASDs can
  be solutions to exactly the same HSDs, and
\item the subset of ASDs that have an injective variable-state mapping
  function,and
\item the construction, by adding equality tests, for every ASD of an
  equivalent ASD in this subset.
\end{ilitz}

\begin{example}{\label{ex:narration}}
  Let us consider the cryptographic protocol for deduction system
  $\DY$ where $\call{F}_\call{D}$ and $\call{P}_\call{D}$ have been
  extended by a free public symbol $f$:
  $$
  \begin{array}{c@{\rightarrow}c@{:}l}
    A & B &  ~\penc{N_a}{\pk B}\\
    B & A & ~\penc{f(N_a) }{\pk A}\\
    \multicolumn 3l {\text{\bf where }}\\
    \multicolumn 3l {A\textbf{ knows }A,B,\pk B,\pk A, \sk A} \\
    \multicolumn 3l {B\textbf{ knows }A,B,\pk A,\pk B, \sk B} \\
  \end{array}
  $$
  Let us define a symbolic derivation for role $B$:
  $$
  \begin{array}{rcl}
    \Ind_B &=& \set{1,\ldots,9}\\ 
    \call{V}_B{}&=& i\in\Ind \mapsto x_i\\
    \call{K}_B{}  &=&\set{A,B,\pk A,\pk B, \sk B}\\
    \Invar_B{} &=&\set{6}\\
    \Outvar_B{} &=&\set{9}\\
    {\call{S}_B}&=&\{x_1\unif{}A,x_2\unif{}B,x_3\unif{}\pk A,x_4\unif{}\pk B, x_5\unif{} \sk B\\
    &&x_7\unif{}\pdec{x_6}{x_5}, x_8\unif{}f(x_7), x_9\unif{} \penc{x_8}{x_3}\}
  \end{array}
  $$
  The set of deduction states in $B$ is $\{7,8,9\}$, there are no
  re-use state, the set of memory states is $\{1,\ldots,5\}$ and the
  only reception state is $6$. Assuming that the role $B$ tests
  whether the received message is a cipher, one may add a tenth
  deduction state with $x_{10}\unif{}\penc{x_7}{x_4}$ and an equation
  $x_6\unif{} x_{10}$.

  Similarly, a symbolic derivation for role $A$ would be:
  $$
  \begin{array}{rcl}
    \Ind_A &=& \set{1,\ldots,10}\\ 
    \call{V}{}&=& i\in\Ind \mapsto y_i\\
    \call{K}{}  &=&\set{A,B,\pk A,\pk B, \sk A,Na}\\
    \Invar{} &=&\set{9}\\
    \Outvar{} &=&\set{7}\\
    {\call{S}}&=&\{y_1\unif{}A,y_2\unif{}B,y_3\unif{}\pk
    A,y_4\unif{}\pk B, y_5\unif{} \sk A,y_6\unif Na\\
    &&y_7\unif{}\penc{y_5}{y_3}, y_8\unif{}f(y_6), y_{10}\unif{}
    \pdec{y_9}{y_5},y_{10}\unif y_8\}
  \end{array}
  $$
  The set of deduction states in $A$ is $\{6,7,9\}$, there are no
  re-use state, the set of memory states is $\{0,\ldots,5\}$ and the
  only reception state is $8$. We have added an equality test
  $y_9\unif y_7$ to model that $A$ checks whether the message received
  actually contains the encryption of $f(Na)$.
  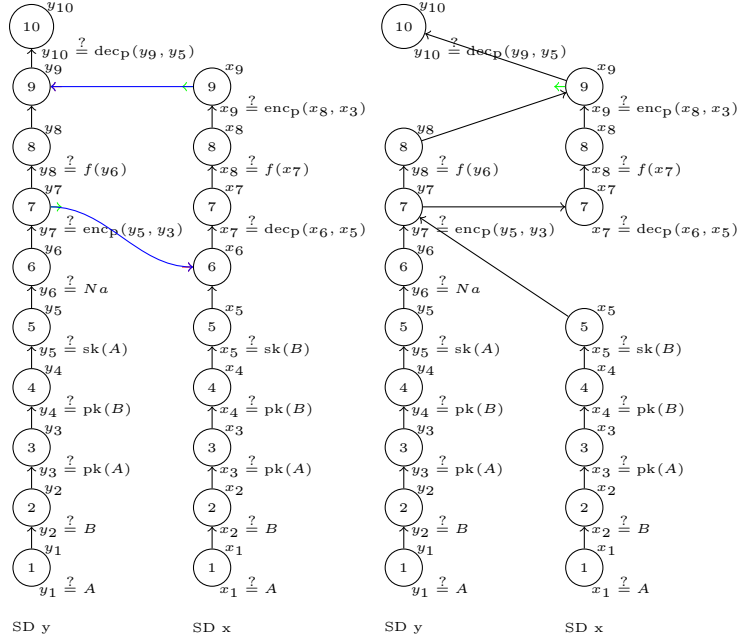
\begin{figure}[htbp]\centering
      \subfloat[With a connection, before computing the resulting
      symbolic derivation]{\hspace*{-1.5cm}
\begin{tikzpicture}[font=\tiny,label distance=-2mm]
        \symbolicderivation{y}{1/d/A,2/d/B,3/d/{\pk{A}},4/d/{\pk{B}},5/d/{\sk{A}},6/d/{Na},7/ed/{\ensuremath{\penc{y_5}{y_3}}},8/d/{\ensuremath{f(y_6)}},9/i/r,10/d/{\ensuremath{\pdec{y_9}{y_5}}}}
        \symbolicderivation{x}{1/d/A,2/d/B,3/d/{\pk{A}},4/d/{\pk{B}},5/d/{\sk{B}},6/i/l,7/d/{\ensuremath{\pdec{x_6}{x_5}}},8/d/{\ensuremath{f(x_7)}},9/wd/{\ensuremath{\penc{x_8}{x_3}}}}
        \connecttoright{sdy-7}{sdx-6}; \connecttoleft{sdx-9}{sdy-9};
      \end{tikzpicture}}
      \subfloat[After computing the
      symbolic derivation resulting from the connection]{\begin{tikzpicture}[font=\tiny,label distance=-2mm]
        \symbolicderivation{y}{1/d/A,2/d/B,3/d/{\pk{A}},4/d/{\pk{B}},5/d/{\sk{A}},6/d/{Na},7/d/{\ensuremath{\penc{y_5}{y_3}}},8/d/{\ensuremath{f(y_6)}},9/h/r,10/d/{\ensuremath{\pdec{y_9}{y_5}}}}
        \symbolicderivation{x}{1/d/A,2/d/B,3/d/{\pk{A}},4/d/{\pk{B}},5/d/{\sk{B}},6/h/l,7/d/{\ensuremath{\pdec{x_6}{x_5}}},8/d/{\ensuremath{f(x_7)}},9/wd/{\ensuremath{\penc{x_8}{x_3}}}}
          \orderto{sdx-5}{sdy-7}
         \orderto{sdy-7}{sdx-7}
         \orderto{sdy-8}{sdx-9}
         \orderto{sdx-9}{sdy-10}
      \end{tikzpicture}\hspace*{-1.5cm}}
    \caption{Honest symbolic derivations of Example~\ref{ex:narration}
      with a connection corresponding to the intended communications
      and the test equations not shown}
\end{figure}
Generally speaking, if ground reachability and ground symbolic
equivalence for the deduction system are decidable (see
Section~\ref{subsec:decision:problems}) then an \emph{as prudent as
  possible} set of deductions and equality tests for the narration can
be computed (see~\cite{cr-ipl}).
\end{example}

In addition we assume that two symbolic derivations do not share any
variable, and that equality between symbolic derivations is defined
modulo a renaming of variables.  The proof of the following lemma is a
direct consequence of the definition.

\begin{lemma}{(Properties of symbolic derivations)\label{lem:ASD:properties}}
  Let $\call{C}=\defsd{}$ be a symbolic derivation. We have:
  \begin{enumerate}{(i)}
  \item For every variable $\call{V}(i)$ there is at most one equation
    in \call{S} of the form $\call{V}(i)\unif{} f(t_1,\ldots,t_n)$;
  \item If $\call{V}(i)$ is a variable such that the above equation is
    in \call{S}, then either
    \begin{ilitz}
    \item $i$ is a deduction state and $i=\min(j\tq \call V(i)=\call
      V(j))$, or
    \item $i$ is a re-use state.
    \end{ilitz}
  \end{enumerate}
\end{lemma}

We rely on the normal form defined by the \textit{o}-completion of the
equational theory \call{E} to prove that every closed symbolic
derivation defines in a unique way the terms deduced.

\begin{lemma}{\label{lem:closed:imply:ground}}
  Let \call{I} be a deduction system, and consider a closed and
  satisfiable \call{I}-symbolic derivation $\call{C}=\defsd{}$.  Then
  there exists a unique ground substitution $\sigma$ in normal form
  that satisfies \call{S}.
\end{lemma}

\begin{proof}
  Since the symbolic derivation $\call{C}=\defsd{}$ is closed is has
  by definition no input states, and thus all states are either
  knowledge, re-use or deduction states.  By induction on the set of
  indexes $\Ind$ ordered by $<$.
  \begin{description}
  \item[Base case:] Assume $i$ is a minimal element in $\Ind$. By
    minimality $i$ cannot be a re-use state. If it is a knowledge
    state then by definition there exists in $\call{S}$ an equation
    $\call{V}(i)\unif t$, with $t$ a ground term in normal form, and
    thus for every unifier $\tau$ of $\call{S}$ we must have
    $\call{V}(i)\tau = t$. If $i$ is a deduction state, and since it
    is minimal, the public symbol employed must be of arity $0$ and
    hence is a constant, \emph{i.e.} again a ground term $t$. In both
    cases there exists a unique ground substitution $\sigma$ in normal
    form defined on \set{\call{V}(i)} and such that any unifier of
    \call{S} is an extension of $\sigma$.
  \item[Induction case:] Assume there exists a unique ground
    substitution $\sigma$ in normal form with support:
    $\condset{\call{V}(j)}{j<i}$ such that any unifier of \call{S} is
    an extension of $\sigma$. If $i$ is a re-use state, we note that
    $\call{V}(i)$ is already in the support of $\sigma$, and we are
    done. If it is a knowledge state, reasoning as in the basic case
    permits us to extend $\sigma$ to $\call{V}(i)$. If it is a
    deduction state then there exists in $\call{S}$ an equation
    $\call{V}(i)\unif f(\call{V}(j_1),\ldots,\call{V}(j_n))$ with
    $j_1,\ldots,j_n < i$ that has to be satisfied by every unifier
    $\theta$ of $\call{S}$. By induction every such unifier has to be
    equal to $\sigma$ on \set{\call{V}(j_1),\ldots,\call{V}(j_n)}.
    Thus for every unifier $\theta$ of $\call{S}$ we have
    $\call{V}(i)\theta=_{\call{E}}
    f(\call{V}(j_1)\theta,\ldots,\call{V}(j_n)\theta)$. By induction
    $f(\call{V}(j_1)\theta,\ldots,\call{V}(j_n)\theta) =_{\call{E}}
    f(\call{V}(j_1)\sigma,\ldots,\call{V}(j_n)\sigma)$. Thus, we have
    $\call{V}(i)\theta =
    \norm{f(\call{V}(j_1)\sigma,\ldots,\call{V}(j_n)\sigma)}$ and
    $\sigma$ can be uniquely extended on $\call{V}(i)$ with
    $\call{V}(i)\sigma =
    \norm{f(\call{V}(j_1)\sigma,\ldots,\call{V}(j_n)\sigma)}$ which is
    again a ground term.
  \end{description}
\end{proof}

By Lemma~\ref{lem:closed:imply:ground}, if a derivation is closed,
then for every $i\in\Ind$ the variable $\call{V}(i)$ is instantiated by a
ground term. Figuratively we  say that a term $t$ is \emph{known at
step $i$} in a closed symbolic derivation if there exists $j\le i$ such
that $\call{V}(j)$ is instantiated by $t$. 

\paragraph{Ground symbolic derivations.}
An important case when considering protocol refutation is the one in
which the attacker cannot alter the messages exchanged among the
honest participants. This case can either be employed to model a
weaker attacker or, when trying to refute a cryptographic protocol, by
guessing first which messages are sent by the attacker, and then by
checking whether these guesses correspond to messages the attacker can
actually send.

\begin{definition}{(Ground symbolic derivation)\label{def:protmod:ground:sd}}
  We say that a symbolic derivation $\call{C}_h=\defsd{_h}$ is a
  \emph{ground symbolic derivation} whenever $\call{S}_h$ is
  satisfiable and there exists a ground substitution $\sigma$ such
  that, for every unifier $\tau$ of $\call{S}_h$ and every
  $i\in\Ind_h$ we have $\V_h(i)\sigma=\V_h(i)\tau$.
\end{definition}

In other words the input and output messages of a ground symbolic
derivation are fixed ground terms. We note that since $\call{C}_h$ is
not closed, and in spite of having $\call{S}_h$ satisfiable, it is not
necessarily true that $\call{C}_h^\star\neq\emptyset$. Also a simple
analysis of the case study of the proof of
Lemma~\ref{lem:closed:imply:ground} shows that it suffices to assume
that $\sigma$ is defined only on indexes $i\in\Invar_h$.

\paragraph{Connection.} We express the communication between two
agents represented each by a symbolic derivation by \emph{connecting}
these symbolic derivations. This operation consists in identifying
some input variables of one derivation with some output variables of
the other and contrariwise.  This connection should be compatible with
the variable orderings inherited from each symbolic derivation, as
detailed in the following definition:

\begin{definition}{\label{def:connection:derivation}}
  Let $\call{C}_1$, $\call{C}_2$ be two symbolic derivations with for
  $i\in\set{1,2}$ $\call{C}_i=\defsd{_i}$, with disjoint sets of variables
  and index sets $(\Ind_1,<_1)$ and $(\Ind_2,<_2)$ respectively.  Let
  $I_1,I_2$, be subsets of $\Invar{}_1$, $\Invar{}_2$, and $O_1,O_2$
  be sub-multisets of $\Outvar{}_1$, $\Outvar{}_2$ respectively.

  Assume that there is a monotone bijection $\phi$ from
  $I_1\cup I_2$ to $O_1\cup O_2$ such that $\phi(I_1)=O_2$ and
  $\phi(I_2)=O_1$.  A \emph{connection} of $\call{C}_1$ and $\call{C}_2$ over the
  \emph{connection} function $\phi$, denoted $\call{C}_1\comp_\phi\call{C}_2$ is a symbolic
  derivation
  $$
  \call{C}=(\call{V},\phi(\call{S}_1 \cup \call{S}_2),\call{K}_1\union\call{K}_2,(\Invar{}_1 \cup
  \Invar_2)\setminus (I_1\cup I_2), (\Outvar_1\cup\Outvar_2)\setminus
  (O_1 \cup O_2))
  $$ 
  where:
  \begin{itemize}
  \item $(\Ind,<)$ is defined by:
    \begin{itemize}
    \item $\Ind=(\Ind_1\setminus I_1) \cup (\Ind_2\setminus I_2)$;
    \item $<$ is the transitive closure of the relation: $<_1 \cup
      <_2$;
    \end{itemize}
  \item $\phi$ is extended to a renaming of variables in
    $\Var{\call{V}_1}\cup\Var{\call{V}_2}$ such that $\phi(\call{V}_1(i))=\call{V}_2(j)$
    (resp.  $\phi(\call{V}_2(i))=\call{V}_1(j)$) if $i\in I_1$ (resp. $I_2$) and
    $\phi(i)=j$
  \end{itemize}
  When the exact connection function in a connection does not matter, is
  uniquely defined, or is described otherwise, we will omit the
  subscript and denote it $\call{C}_1\comp \call{C}_2$.
\end{definition}

A connection is \emph{satisfiable} if the resulting symbolic
derivation is satisfiable. It can easily computed, when it exists, by
considering increasing sequences of states in each symbolic derivation
and mapping input states of one SD with output states of the other.

\begin{example}{\label{ex:compos}}
  Let $\call{C}_h$ be the symbolic derivation in
  Example~\ref{ex:narration}:
  $$
  \begin{array}{rcl}
    \Ind_h &=& \set{0,\ldots,8}\\ 
    \call{V}{}_h&=& i\in\Ind \mapsto x_i\\
    \call{K}{}_h  &=&\set{A,B,\pk A,\pk B, \sk B}\\
    \Invar{}_h &=&\set{5}\\
    \Outvar{}_h &=&\set{0,\ldots,8,8}\\
    {\call{S}}_h&=&\{x_0\unif{}A,x_1\unif{}B,x_2\unif{}\pk A,x_3\unif{}\pk B, x_4\unif{} \sk B\\
    &&x_6\unif{}\pdec{x_5}{x_4}, x_7\unif{}f(x_6), x_8\unif{} \penc{x_7}{x_2}\}
  \end{array}
  $$
  We model the initial knowledge of the intruder with another symbolic
  derivation $\call{C}_K$:
  $$
  \begin{array}{rcl}
    \Ind_K &=& \set{0^k,\ldots,3^k}\\ 
    \call{V}_K&=& i^k\in\Ind_k \mapsto y_i\\
    \call{K}{}_K  &=&\set{A,B,\pk A,\pk B}\\
    \Invar{}_K &=&\emptyset\\
    \Outvar{}_K &=&\Ind_K\\
    {\call{S}}_K&=&\set{y_{0}\unif{}A,y_{1}\unif{}B,y_{2}\unif{}\pk A,y_3\unif{}\pk B}
  \end{array}
  $$
  and we let $\call{C}'$ be the following derivation:
  $$
  \begin{array}{rcl}
    \Ind' &=& \set{0',\ldots,8} \\ 
    \call{V}'{} &=& i'\in\Ind' \mapsto z_i\\
    \call{K}{} &=&\set{n} \subset \Nonces\\
    \Invar{}'  &=&\set{0',\ldots,3',8'}\\
    \Outvar{}' &=&\set{5'}\cup\Ind'\\
    {\call{S}}'      &=&\{z_{4}\unif{}n,z_{5}\unif{}\penc{z_{4}}{z_{3}}, \\
    & & z_{6}\unif{}f(z_{4}), z_{7} \unif{} \penc{z_6}{z_2},z_8\unif{}z_7\}
  \end{array}
  $$
  Let $\phi$ be the application from $0^k,\ldots, 3^k, 5',8$ to
  $0',\ldots, 3', 5, 8'$ respectively and $\psi$ be a function of
  empty domain.  Then we have $(\call{C}_h\comp_\psi \call{C}_K)
  \comp_{\phi} \call{C}'$:
  $$
  \begin{array}{rcl}
    \Ind &=& \set{0,\ldots,4,0^k,\ldots,3^k,5',6',7',6,7,8}\\ 
    \call{V}{} &=& {\call{V}_h}_{|\Ind} \cup {\call{V}_K}_{|\Ind}\cup {\call{V}'}_{|\Ind}\\
    \call{K}{} &=&\set{A,B,\pk A, \pk B, \sk B,n} \\
    \Invar{}  &=&\emptyset\\
    \Outvar{} &=&\Ind\cap\Ind'\\
    {\call{S}}      &=&\{x_0\unif{}A,x_1\unif{}B,x_2\unif{}\pk A,x_3\unif{}\pk B, x_4\unif{} \sk B\\
    &&x_6\unif{}\pdec{x_5}{x_4}, x_7\unif{}f(x_6), x_8\unif{} \penc{x_7}{x_2}\\
    &&y_{0}\unif{}A,y_{1}\unif{}B,y_{2}\unif{}\pk A,y_3\unif{}\pk B\\
    && z_{5}\unif{}n,z_{6}\unif{}\penc{z_{5}}{z_{3}},\\
    && z_{7}\unif{}f(z_{5}), z_{8} \unif{} \penc{z_7}{z_2},z_9\unif{}z_8\}
  \end{array}
  $$
  with the ordering:
  $$
  \begin{array}{l}
    0 < 1 < 2 < 3 < 4 < 5' < 6 < 7 < 8\\
    0^k < \ldots < 3^k <  4' < \ldots < 7'<8\\
  \end{array}
  $$
\end{example}

The connection of two symbolic derivations $\call{C}_1$ and
$\call{C}_2$ identifies variables in the input of one with variables
in the output of the other. Variables that have been identified are
removed from the input/output set of the resulting symbolic derivation
\call{C}. The set of equality constraints of \call{C} is the union of
the equality constraints in $\call{C}_1$ and $\call{C}_2$, plus
equalities stemming from the identification of input and output. We
have chosen to have a multiset of output variables to enable the
modeler to specify whether a communication between two participants is
hidden---when the output state occurs only once in the initial output
multiset---or visible---in which case there is more than one
occurrence of the output state in the initial output multiset---to an
external observer.

One easily checks that a connection of two symbolic derivations is
also a symbolic derivation. Also, the associativity of function
composition applied on the connections implies the associativity of
the connection of symbolic derivations. Since connection functions are
bijective, we will also identify $\call{C}\comp\call{C}'$ and $\call{C}'\comp\call{C}$. Thus
when we compose several symbolic derivations, we will freely
re-arrange or remove parentheses.

\paragraph{Traces.} Let $\call{C}_1$ and $\call{C}_2$ be two
\call{I}-symbolic derivations and $\varphi$ be a connection such that
$\call{C}=\call{C}_1\comp_\varphi\call{C}_2=\defsd{}$ is closed and
satisfiable.  Lemma~\ref{lem:closed:imply:ground} implies that there
exists a unique ground substitution $\tau$ in normal form such that
any unifier $\sigma$ of $\call{S}_1\cup\call{S}_2$ is equal to $\tau$
on the image of \call{V}. We denote
\trace{\call{C}'}{\call{C}_1\comp_\varphi\call{C}_2} the restriction
of this substitution $\tau$ to the variables in the sequence of
$\call{C}'$, for
$\call{C}'\in\set{\call{C}_1,\call{C}_2,\call{C}_1\comp_\varphi\call{C}_2}$,
and call it the \emph{trace} of the connection on $\call{C}'$. In the
rest of this paper we will always assume that trace substitutions are
in normal form.

\subsection{Solutions of symbolic derivations}
\label{subsec:symb:sat}

\subsubsection{Honest  and attacker symbolic derivations}

Generally speaking, a \emph{solution} of a symbolic derivation
$\call{C}$ is any couple $(\call{C}',\varphi)$ such that
$\call{C}\comp_\varphi\call{C}'$ is closed and satisfiable. We
specialize this definition for the case of protocol analysis in order
to ensure that every term possessed by the attacker, including her
initial knowledge, has been either leaked by the protocol or is a
nonce she has created. This consideration lead us to consider two
types of symbolic derivations, one that is employed to model honest
agents, and one to model an attacker.

\paragraph{Honest derivations.}
We do not impose constraints on the symbolic derivations representing
honest principals, but for the avoidance of constants in an infinite
set $\Nonces\subseteq \Constants$. These constants are employed to
model new values created by an attacker. We assume that nonces created
by the honest agents are created at the beginning of their execution
and are constants away from \Nonces.

\begin{definition}{(Honest symbolic derivations)}
  A symbolic derivation $\call{C}$ is an \emph{honest symbolic
    derivation} or HSD, if the constants occurring in \call{C}
  are away from \Nonces.
\end{definition}

  \begin{example}{\label{ex:hsd}}
    The symbolic derivation for role B in Example~\ref{ex:narration}
    is honest.
  \end{example}

\paragraph{Attacker derivations.}
We consider an attacker modeled by a symbolic derivation in which 
only the following actions are possible:
\begin{itemize}
\item create a fresh, random value; 
\item receive from and send a message to  one of the honest participant;
\item deduce a new message from the set of already known messages;
\item every state is in \Outvar{} given that the intruder should be
  able to observe his own knowledge;
\item given that we consider an actual execution, the set of states is
  totally ordered.
\end{itemize}
The definition of \emph{attacker} symbolic derivations models these
constraints:

\begin{definition}{\label{def:minimal:sd}(Attacker symbolic derivations)}
  Let $\call{C}=\defsd{}$ be a symbolic derivation. It is an
  \emph{attacker symbolic derivation}, or ASD, if
  \begin{ilitz}
  \item \Ind is a total order, and
  \item \Outvar contains at least one occurrence of each index in
    \Ind, and
  \item \call{K} is a subset of \Nonces.
  \end{ilitz}
\end{definition}

The fact that the initial knowledge of the attacker is empty but for
the nonces is not a restriction when analyzing protocols, as one can
see from Ex.~\ref{ex:compos}.

\begin{example}{\label{ex:asd}}
  The following derivation $\call{C}'$ is an ASD for the same
  deduction system as Example~\ref{ex:narration}:
  $$
  \begin{array}{rcl}
    \Ind' &=& \set{0',\ldots,8} \\ 
    \call{V}'{} &=& i'\in\Ind' \mapsto z_i\\
    \call{K}{} &=&\set{n} \subset \Nonces\\
    \Invar{}'  &=&\set{0',\ldots,3',8'}\\
    \Outvar{}' &=&\set{5'}\cup\Ind'\\
    {\call{S}}'      &=&\{z_{4}\unif{}n,z_{5}\unif{}\penc{z_{4}}{z_{3}}, \\
    & & z_{6}\unif{}f(z_{4}), z_{7} \unif{} \penc{z_6}{z_2},z_8\unif{}z_7\}
  \end{array}
  $$
  Informally the ASD expresses that the attacker receives some key
  $k$, creates a nonce $n$, sends the encrypted nonce to a role $B$ as
  in Example~\ref{ex:narration}.  Then the attacker tries to check
  that applying $f$ to $n$ gives a term equal to the decryption of B's
  response.
\end{example}

\paragraph{Solutions of a symbolic derivation.}
Given a symbolic derivation $\call{C}_h$ we denote $\call{C}_h^\star$
the set of couples $(\call{C},\varphi)$ where $\call{C}$ is an ASD and
$\varphi$ is a connection function between $\call{C}$ and $\call{C}_h$
such that $\call{C}_h\comp\call{C}$ is closed and satisfiable. In that
case we say that \call{C} is a solution of $\call{C}_h$.

\begin{example}{\label{ex:solasd}}
  In Example~\ref{ex:compos} the ASD $\call{C}'$ is a solution of
  $\call{C}_h\comp\call{C}_K$ since
  $(\call{C}_h\comp_\psi\call{C}_K)\comp_{\phi} \call{C}'$ is closed
  and ${\call{S}}$ is satisfiable (by simply propagating the
  equalities $x_0=A,x_1=B,\ldots$).
\end{example}

\subsection{Decision problems}
\label{subsec:decision:problems}

\paragraph{Satisfiability.} The problem of the existence of a secrecy
attack on a bounded protocol execution---shown to be NP-complete
in~\cite{RT01} for the standard Dolev-Yao deduction system---is
equivalent to the satisfiability problem below.

\begin{decisionproblem}{\call{I}-Satisfiability}
  \entreeu{a HSD \call{C}} %
  \sortie{\textsc{Sat} iff $\call{C}^\star\neq\emptyset$}
\end{decisionproblem}

A variant of \call{I}-satisfiability is its restriction to set of
inputs $\call{C}$ which are ground symbolic derivations, and that we
call \call{I}-ground satisfiability.
\begin{decisionproblem}{Ground \call{I}-Satisfiability}
  \entreeu{a ground HSD \call{C}} %
  \sortie{\textsc{Sat} iff $\call{C}^\star\neq\emptyset$}
\end{decisionproblem}

\paragraph{Equivalence.} Let us now define the equivalence of HSDs
\textit{w.r.t.}  an active intruder.
\begin{definition}{\label{def:symbolic:equivalence}}
  Two HSDs $\call{C}_h$ and $\call{C}_h'$ are \emph{symbolically
    equivalent} iff $\call{C}_h^\star={\call{C}_h'}^\star$.
\end{definition}

\begin{decisionproblem}{\call{I}-Symbolic Equivalence}
  \entreeu{Two honest \call{I}-symbolic derivations $\call{C}_h$ and $\call{C}_h'$} %
  \sortie{\textsc{Sat} iff ${\call{C}_h}^\star={\call{C}_h'}^\star$.}
\end{decisionproblem}

Again it is possible to define a ground version of the
\call{I}-symbolic equivalence problem when the input consists in two
ground symbolic derivations. One can easily encode static equivalence
problems into ground $\call{I}$-Symbolic Equivalence problems by
publishing every constant not hidden in the frame.
\begin{decisionproblem}{Ground \call{I}-Symbolic Equivalence}
  \entreeu{Two honest \call{I}-ground symbolic derivations
    $\call{C}_h$ and $\call{C}_h'$} %
  \sortie{\textsc{Sat} iff ${\call{C}_h}^\star={\call{C}_h'}^\star$.}
\end{decisionproblem}

\paragraph{Remark.} Another possible definition of the set of
solutions would be a set of ASDs, without mention of the connection
function.  The equivalence relation would have been distinct since in
that case an ASD can be in two sets of solutions but without the same
connection function. However, this would have had no impact on our
decidability result. Our choice in this paper corresponds to
diff-equivalence between
biprocesses~\cite{BlanchetAbadiFournetLICS05}: the diff operator
defines a bijection between the in- and output states of two processes
derivations, and the equality of the sets of solutions is understood
modulo this one-to-one function.

\section{Finitary Deduction Systems}
\label{equiv:sec:finitary}

An equational theory \call{E} is \emph{finitary} whenever every
\call{E}-unification system has a finite set of more general unifiers.
We define an analog for deduction systems \textit{w.r.t.} symbolic
derivations rather than  equational theories \textit{w.r.t.}
unification systems. In the rest of this paper, we consider
\emph{effective} finitary deduction systems, \textit{i.e.} deduction
systems for which it is possible to compute a finite set of ``most
general attacks''.

\subsection{Stutter-free ASDs}
\label{equiv:subsec:aware:stutter:free}

We say that an ASD $\call{C_I}$ is \emph{well-formed} \textit{w.r.t.}
a HSD $\call{C}_h$ and a connection $\varphi$ if, in the connection
$\call{C}_h\comp_\varphi\call{C_I}$, a deduction subsequently applied
on a deduced term $t$, or a re-use of the term $t$ is always applied
by referring to the state in which $t$ was first deduced.

\begin{definition}{(Well-formed ASD)\label{def:well:formed}}
  Let $\call{C}_h$ be a HSD and consider an ASD $\call{C_I} = \defsd{_\call{I}}$ 
  such that $(\call{C_I},\varphi)\in\call{C}_h^\star$, and
  $\sigma=\trace{\call{C_I}}{\call{C_I}\comp_\varphi\call{C}_h}$. We
  say that $\call{C_I}$ is \emph{$(\call{C}_h,\varphi)$-well-formed}
  if for every deduction states $i$, for every state
  $j\in\Ind_{\call{I}}$ with $i<j$ we have
  $\call{V_I}(i)\sigma=\call{V_I}(j)\sigma$ implies that
  \begin{itemize}
  \item either $\call{V_I}(i)=\call{V_I}(j)$, \textit{i.e.} $j$ is a
    re-use state;
  \item or there is no equation $x\unif
    f(\ldots,\call{V_I}(j),\ldots)$ in $\call{S_I}$ and $j$ is not an
    emission state.
  \end{itemize}
\end{definition}

This restriction is mostly syntactic, and can be assumed
\textit{w.l.o.g.} for our purpose, as shown by the
Lemma~\ref{lem:well:formed}.

Our aim is the reduction of equivalence problems to reachability
problems for finitary deduction systems. In the latter problems, one
only considers which terms are deducible by the attacker. Hence the
following definitions that will be employed to split an ASD into a
\emph{deduction only} part solving a reachability problem and a
\emph{testing} part modeling the possible tests.

\begin{definition}{(Deduction-only ASD)\label{def:deduction:only}}
  An ASD $\call{C_I}=\defsd{_{\call{I}}}$ is \emph{deduction-only} if
  $\call{S_I}$ contains no equation $\call{V_I}(i)\unif\call{V_I}(j)$.
\end{definition}

\begin{definition}{(Testing ASD)\label{def:testing}}
  An ASD $\call{C_I}=\defsd{_{\call{I}}}$ is \emph{testing} if
  $\call{K_I}=\emptyset$.
\end{definition}

\begin{definition}{(Stutter-free ASDs)\label{def:stutter:free}}
  A well-formed deduction-only ASD is said to be \emph{stutter-free}.
\end{definition}

Given a HSD $\call{C}_h$ we denote ${\call{C}_h}^\sfree$ the set of
stutter-free solutions of $\call{C}_h$. These ASDs have the special
property that a connection cannot be unsatisfiable because of a
rejection by the attacker. Formally speaking, we have the following
proposition.

\begin{proposition}{\label{equiv:prop:stutter:free}}
  Let $\call{C}_\call{I}=\defsd{_\call{I}}\in\call{C}_h^\star$ be a
  deduction-only ASD. Then for any ground substitution $\sigma$ of
  domain $\Invar_\call{I}$ the unification system
  $\call{S}_\call{I}\sigma$ is satisfiable in the empty theory.
\end{proposition}

\begin{proof}
  We remind that a unification system $\call{S}$ is in solved form in
  the empty theory if and only if there exists an ordering $<_u$ on
  variables such that $\call{S}$ contains, for each variable $x$, at
  most one equation $x\unif{}t$ and if for every $y\in\Var{t}$ we have
  $y<_u x$.  First let us notice that since $\call{C}_\call{I}$ is
  deduction-only, $\call{S}_\call{I}$ does not contain any equation
  $\call{V}_\call{I}(i)\unif \call{V}_\call{I}(j)$ with
  $\call{V}_\call{I}(i)\neq \call{V}_\call{I}(j)$.

  By definition $\call{S}_\call{I}$ contains exactly one equation
  $\call{V}_\call{I}(i)\unif t$ if $i$ is not an input or the re-use
  of an input state, and none otherwise.  In the former case we can
  assume that for a mgu $\theta$ of $\call{S}$ we have
  $\call{V}(i)\theta=\call{V}(i)$. Using the ordering on states as the
  ordering $<_u$, Lemma~\ref{lem:ASD:properties} implies that
  $\call{S}_\call{I}$ is in solved form, and adding to
  $\call{S}_\call{I}$ equations $\call{V}_\call{I}(i)\unif t_i$, for
  $i\in\Invar_\call{I}$ and $t_i$ a ground term thus leads to a
  unification system also in solved form.
\end{proof}

\subsection{Sets of solutions}

\paragraph{Outline.}
We prove in this section that ASDs are such that, when replacing a
constant in \Nonces{} by the result of a sequence of compositions
(this operation is called \emph{opening}) we obtain another ASD which
can be connected to all the HSDs the original ASD could be connected
to (Lemma~\ref{equiv:lem:replace:constant:by:term}). This notion of
replacement acts as the instantiation of a unifier modulo an
equational theory. Accordingly we define from it a well-founded
ordering on ASDs mimicking the role of the instantiation ordering on
unifiers. Finally, we prove that given a set of ASDs $S$, the
inclusion $S\subseteq \call{C}_h^\star$ can be check by testing only
the minimal ASDs in $S$ (Lemma~\ref{equiv:lem:reduction:sol}).

\paragraph{Opening of symbolic derivations.} If $\call{C}=\defsd{}$
and $C\subseteq \Nonces\cap \call{K}$ is a set such such that
$C\cap\Sub{\call{K}\setminus C}=\emptyset$, we \emph{open \call{C} on
  $C$}, and denote the operation $\open{\call{C}}{C}$, when for each
$c\in C$:
\begin{itemize}
\item If $i\in\Ind$ is the first knowledge state with
  $\call{V}(i)\unif c\in\call{S}$, we remove this equation from
  $\call{S}$ and add $i$ to the input states;
\item we replace all occurrences of $c$ in \call{C} by $\call{V}(i)$.
\end{itemize}
We note that the set $\call{K}'$ obtained from \call{K} after the
replacement is still a set of ground terms since
$C\cap\Sub{\call{K}\setminus C}=\emptyset$, and thus the result of the
operation is still a symbolic derivation. Also, $\call{C}$ is an ASD,
then so is $\open{\call{C}}{C}$.

\begin{lemma}{\label{equiv:lem:replace:constant:by:term}}
  Let $\call{C}_\call{I}\in\call{C}_h^\star$ with
  $\call{C}_\call{I}=\defsd{_\call{I}}$, let
  $C\subseteq\call{K}_\call{I}$ and let
  $\call{C}_c\in{\call{C}_h'}^\sfree$ for some HSD $\call{C}_h'$. If a
  connection $\call{C}_c\comp
  \call{C}_h\comp\open{\call{C}_\call{I}}{C}$ is closed then it is
  satisfiable.
\end{lemma}

\begin{proof}
  By Proposition~\ref{equiv:prop:stutter:free} 
  $\trace{\call{C}_c}{\call{C}_c\comp\call{C}_h\comp\open{\call{C}_\call{I}}{\set{c}}}$
  satisfies $\call{S}_c$. Since $\call{C}_\call{I}$ is an ASD we have
  $C\cap\Sub{\call{K}\setminus C}=\emptyset$, and thus
  $C\cap\Sub{\call{S}_h}=\emptyset$. Let us denote
  $\call{S}_\call{I}'$ the unification system $\call{S}_\call{I}$ in
  which the equations $x\unif c$ with $c\in C$ are removed. For any
  substitution $\sigma$ and any constant $c\in C$,
  Lemma~\ref{lem:unif:replacement} and
  $\sigma\models_{\call{E}}\call{S}_h\comp\call{S}_\call{I}'$ imply
  $\sigma\delta_{c,t}\models_{\call{E}}\call{S}_h\comp\call{S}_\call{I}'$.

  Let
  $\sigma'=\trace{\call{C}_\call{I}}{\call{C}_c\comp\call{C}_h\comp\open{\call{C}_\call{I}}{C}}$.
  For each memory state $i\in\Ind_\call{I}$ that contains a constant
  $c\in C$ we let $t_c=\call{V}_\call{I}(i)\sigma'$. We define
  $\delta$ as the replacement of each constant $c\in C$ by the term
  $t_c$.

  By induction on the indexes of the connection
  $\call{C}_c\comp\call{C}_h\comp\open{\call{C}_\call{I}}{C}$ we have:
  $$
  \trace{\call{C}_c\comp\call{C}_h\comp\open{\call{C}_\call{I}}{C}}
  {\call{C}_c\comp\call{C}_h\comp\open{\call{C}_\call{I}}{C}} =
  \trace{\call{C}_h\comp\call{C}_\call{I}}{\call{C}_h\comp\call{C}_\call{I}}\delta
  $$
  Thus every equation in $\call{S}_h\cup\call{S}_\call{I}$ (minus the
  removed memory equations) is satisfied by the composition with
  $\call{C}_c$.  Since every equation in its unification system is
  satisfied the connection
  $\call{C}_c\comp\call{C}_h\comp\open{\call{C}_\call{I}}{C}$ is
  satisfiable.
\end{proof}

\paragraph{Ordering on symbolic derivations.} Consider two symbolic
derivations:
$$
\left\lbrace
  \begin{array}{rcl}
    \call{C}_\call{I}&=&\defsd{_\call{I}}\\
    \call{C}_\call{I}'&=&\defsd{_\call{I}'}\\
  \end{array}
\right.
$$
We say that $\call{C}_\call{I}\le\call{C}_\call{I}'$ if:
\begin{itemize}
\item there exists $C\subseteq\call{K}_\call{I}$, a stutter-free
  symbolic derivation $\call{C}_C$ and a connection $\varphi$ such
  that
  $\call{C}_C\comp_\varphi\open{\call{C}_\call{I}}{C}={\call{C}_\call{I}'}$
  modulo a renaming of variables;
\item or there exists a set of memory states
  $I\subseteq\Ind_\call{I}'$ such that $C_\call{I}$ is equal to
  $\call{C}_\call{I}''=\defsd{_\call{I}''}$ where:
  \begin{itemize}
  \item $\call{V}_\call{I}''$ is the restriction of
    $\call{V}_\call{I}'$ to the domain $\Ind_\call{I}'\setminus I$
  \item and
    $\call{S}_\call{I}''=\call{S}_\call{I}'\setminus\set{\call{V}_\call{I}'(i)\unif
      c_i}_{i\in I}$.
  \end{itemize}
\end{itemize}
We say that $\call{C}_\call{I},\call{C}_\call{I}'$ are
\emph{equivalent modulo a renaming of nonces}, and denote
$\call{C}_\call{I} \equiv \call{C}_\call{I}'$, whenever there exists
$C\subseteq\call{K}_\call{I}$, a stutter-free symbolic derivation
$\call{C}_C$ \emph{with only memory states}, and a connection
$\varphi$ such that
$\call{C}_C\comp_\varphi\open{\call{C}_\call{I}}{C}={\call{C}_h'}$.
Given a set $S$ of ASDs we denote $\min_<(S)$ the set of ASDs in $S$
that are minimal in $S$ modulo renaming of nonces.

Since $\call{C}\le\call{C}'$ implies that either:
\begin{ilitz}
\item $\call{C}$ has strictly less deduction states than $\call{C}'$,
  and less states,
\item \call{C} has strictly less states than \call{C}',
\item or $\call{C}$ and $\call{C}'$ are equivalent modulo a renaming
  of nonces,
\end{ilitz}
it is clear that $<$ is a well-founded ordering relation modulo this
renaming.

\begin{lemma}{\label{equiv:lem:reduction:sol}}
  Let $S$ be a set of ASDs and $\call{C}_h$ be a HSD. If
  $\min_<(S)\subseteq {\call{C}_h}^\star$ then $S\subseteq
  {\call{C}_h}^\star$.
\end{lemma}

\begin{proof}
  Assume $\min_<(S)\subseteq {\call{C}_h}^\star$ and let
  $\call{C}_\call{I}$ be in $S$. By definition of the ordering, first
  point, there exists a derivation $\call{C}_\call{I}'\in \min_<(S)$,
  a set of constants $C$, and a stutter-free derivation $\call{C}_c$
  such that
  $\call{C}_c\comp\open{\call{C}_\call{I}'}{C}=\call{C}_\call{I}$.  By
  hypothesis we have $\call{C}_\call{I}'\in{\call{C}_h}^\star$. By
  Lemma~\ref{equiv:lem:replace:constant:by:term} this implies that
  $\call{C}_\call{I}=\call{C}_c\comp\open{\call{C}_\call{I}'}{C}$ is
  also in ${\call{C}_h}^\star$.
\end{proof}

\paragraph{Complete sets of solutions.}
The ordering $<$ plays the same role \textit{w.r.t.} the solutions of
a HSD as the instantiation ordering on substitutions \textit{w.r.t.}
the solutions of an unification system. In particular the traditional
notion of most general unifier is translated into a notion of minimal
solution.

\begin{definition}{\label{equiv:def:complete:set}(Complete set of solutions)}
  A set $\Sigma$ of ASDs is a \emph{complete set of solutions} of an
  HSD $\call{C}_h$ whenever:
  \begin{itemize}
  \item $\Sigma\subseteq\call{C}_h^\star$;
  \item for every ASD $\call{C}_\call{I}\in\call{C}_h^\sfree$ there
    exists an ASD $\call{C}_m\in\Sigma$ and a stutter free ASD
    $\call{C}_c$ such that
    $\call{C}_m\le\call{C}_\call{I}\comp\call{C}_c$.
  \end{itemize}
\end{definition}

We have departed from our line of translating terms from the
unification framework to the symbolic derivation framework by
introducing a symbolic derivation $\call{C}_c$. It permits us to
consider cases in which the computation of a complete set of unifiers
introduces unnecessary deduction steps in individual ASDs. A common
example of such addition is the normalization of messages
$\paire{t}{t'}$, \textit{i.e.}  the automatic deduction of the two
messages $t$ and $t'$ even when they are not useful for the attacker.

\subsection{Finitary deduction systems}

We have already noted that a NP decision procedure for the
satisfiability of HSDs for the Dolev-Yao deduction system is known
since~\cite{RT01}. While this procedure is based on the guessing of an
attack of minimal size, other procedures have been
proposed~\cite{AmadioL00,MS01} that instead cover all possible
stutter-free derivations~\cite{ChevalierLR-LPAR07}, \textit{i.e.}
compute a complete set of solutions. We define deduction systems for
which such a procedure exists to be \emph{finitary}.

\begin{definition}{\label{equiv:def:finitary}(Finitary Deduction Systems)}
  Let \call{I}{} be a deduction system. If there exists a procedure
  that computes for every \call{I}{}-HSD $\call{C}_h$ a finite
  complete set of solutions we say that \call{I}{} is a \emph{finitary
    deduction system}.
\end{definition}

\section{Decidability of Symbolic Equivalence}
\label{equiv:sec:equiv}

This section is devoted to the proof of the main theorem of this
paper.
\begin{theorem}{\label{equiv:theo:dy:decidable}}
  Symbolic equivalence is decidable for finitary deduction systems.
\end{theorem}

We first prove that every ASD can be written as the connection between
a stutter-free ASD and a testing ASD in which no new term is deduced
(Lemma~\ref{equiv:lem:decompose:ASD:connection}). This implies the
reduction of the inclusion problem to the one of checking whether, for
any stutter-free ASD in $\call{C}_h^\star$, the connections of this
ASD with $\call{C}_h$ and $\call{C}_h'$ result in \emph{closed}
symbolic derivations $\call{C}_1$ and $\call{C}_2$ such that
$\call{C}_1^\star\subseteq \call{C}_2^\star$
(Lemma~\ref{equiv:lem:split}). Given a stutter-free ASD in
$\call{C}_h^\star$ this latter test is simple since it suffices to
consider the connection with ASD that have at most one deduction
(Prop.~\ref{equiv:prop:minimal:test}).

We relate these types of ASD with well-formed ASDs with the following
lemma.

\begin{lemma}{\label{equiv:lem:decompose:ASD:connection}}
  Let $\call{C_I}$ be a $(\call{C}_h,\varphi)$-well-formed ASD. Then
  there exists a connection $\psi$, a well-formed deduction-only ASD
  $\call{C}_d$, and a testing ASD $\call{C}_t$ such that:
  \begin{itemize}
  \item $\call{C_I}=\call{C}_d \comp_\psi \call{C}_t$,
  \item for all HSD $\call{C}'$ and connection $\psi$, the connection
    $\call{C}'\comp_\psi \call{C_I}$ is closed if, and only if,
    $\call{C}'\comp_\psi \call{C}_d$ is closed.
  \end{itemize}
\end{lemma}

\begin{proof}
  Let $\call{C}_h$ be a HSD, $\call{C_I}=\defsd{_{\call{I}}}$ bean
  ASD, and $\varphi$ be a connection such that
  $(\call{C_I},\varphi)\in \call{C}_h^\star$.  We construct a sequence
  of couples $(\call{C}_d,\call{C}_t)$ of ASDs such that $\call{C}_d$
  is deduction-only, $\call{C}_t$ is testing, and such that in the end
  $\call{C}_d$ is well-formed.  We start from:
  $$
  \left\lbrace
    \begin{array}{rcl}
      \call{C}_d&=&(\call{V_I},\call{S_I}\setminus
      \call{S}_=,\call{K_I},\Invar_{\call I},\Outvar_{\call I}\cup
      \Ind_{\call I})\\
      \call{C}_t&=&(\call{V_I},\call{S}_=,\emptyset,
      \Ind_{\call I},\Outvar_{\call I})\\
    \end{array}
  \right.
  $$
  and the connection $\psi$ being the identity. By construction
  $\call{C}_t$ is testing and $\call{C}_d$ is deduction-only.  However
  $\call{C}_d$ may not be well-formed.

  For each deduction state $i$ in $\call{C}_d$ such that there exists
  a deduction state $j<i$ with $\call{V_I}(i)\sigma =
  \call{V_I}(j)\sigma$, let $\call{S}_{\call{V_I}(i)}$ be the subset
  of equations of $\call{S_I}$ in which $\call{V_I}(i)$ occurs.  Since
  $i$ is a deduction state, $\call{S}_{\call{V_I}(i)}$ contains one
  equation $\call{V_I}(i) \unif f(x_1,\ldots,x_n)$.  Since the ASD is
  well-formed, all other equations in $\call{S}_{\call{V_I}(i)}$ are
  of the form $\call{V_I}(i) \unif \call{V_I}(j)$, and thus are
  already in $\call{S}_=$. We obtain a new couple of ASDs
  $(\call{C}_d',\call{C}_t')$ by removing the state $i$ from
  $\call{C}_d$ (and thus from the output variables of $\call{C}_d$,
  removing $i$ from the input states of $\call{C}_t$, and adding the
  equation $\call{V_I}(i) \unif f(x_1,\ldots,x_n)$ to the unification
  system of $\call{C}_t$, thereby making $i$ a deduction state in
  $\call{C}_t$.

  It is clear that once the construction is performed on every
  deduction states from $\call{C}_d$, this symbolic derivation will be
  well-formed.
\end{proof}

\begin{lemma}{\label{lem:well:formed}}
  Let $\call{C}_h,\call{C}_h'$ be two HSDs such that
  $\call{C}_h^\star\setminus {\call{C}_h'}^\star\neq\emptyset$. Then
  $\call{C}_h^\star\setminus {\call{C}_h'}^\star$ contains a
  $(\call{C}_h,\varphi)$-well-formed ASD.
\end{lemma}
\begin{proof}
  Assume
  $(\call{C_I},\varphi)\in\call{C}_h^\star\setminus{\call{C}_h'}^\star$,
  and $\call{C_I}=\defsd{_{\call{I}}}$, and
  $\sigma=\trace{\call{C_I}}{\call{C_I}\comp_\varphi\call{C}_h}$.  By
  hypothesis $\sigma$ satisfies $\call{S_I}$. Let $\call{S}_1$ be the
  set of equations $\call{V_I}(i)\unif\call{V_I}(j)$ on all states
  $i,j$ such that:
  \begin{ilitz}
  \item $i$ is a deduction state, and
  \item $i<j$, and
  \item $\call{V_I}(i)\sigma=\call{V_I}(j)\sigma$.
  \end{ilitz}
  It is clear that $\call{S_I}\cup\call{S}_1$ is also satisfied by
  $\sigma$.

  Then, replace in \call{S_I} each equation $x\unif
  f(\ldots,\call{V_I}(j),\ldots)$ such that there exists a deduction
  state $i<j$ with $\call{V_I}(i)\sigma=\call{V_I}(j)\sigma$ by the
  equation $x\unif f(\ldots,\call{V_I}(i),\ldots)$, and let
  $\call{S_I}'$ be the obtained unification system. Given the
  equations in $\call{S}_1$ it is clear that
  $\call{S_I}\cup\call{S}_1$ and $\call{S_I}'\cup\call{S}_1$ are
  satisfied by the same set of substitutions.

  Let $\call{C_I}'=(\call {V_I},\call {S_I}'\cup\call{S}_1, \call
  {K_I}, \Invar_{\call I}, \Outvar_{\call{I}})$. It remains to note
  that:
  \begin{itemize}
  \item
    $\trace{\call{C_I}\comp_\varphi\call{C}_h}{\call{C_I}\comp_\varphi\call{C}_h}=
    \trace{\call{C_I}'\comp_\varphi\call{C}_h}{\call{C_I}'\comp_\varphi\call{C}_h}$;
  \item
    $\trace{\call{C_I}\comp_\varphi\call{C}_h'}{\call{C_I}\comp_\varphi\call{C}_h'}=
    \trace{\call{C_I}'\comp_\varphi\call{C}_h'}{\call{C_I}'\comp_\varphi\call{C}_h'}$,
    and thus $(\call{C_I}',\varphi)\notin\call{C}_h'$;
  \item by construction $\call{C_I}'$ is
    $(\call{C}_h,\varphi)$-well-formed.
  \end{itemize}
  Thus, $\call{C_I}'$ is $(\call{C}_h,\varphi)$-well-formed ASD in
  $\call{C}_h^\star\setminus {\call{C}_h'}^\star$.
\end{proof}

As a consequence, we obtain the following lemma that permits to split
the symbolic equivalence problem into two simpler problems.

\begin{lemma}{\label{equiv:lem:split}}
  Let $\call{C}_h$ and $\call{C}_h'$ be two HSDs. We have
  $\call{C}_h^\star\subseteq\call{C}_h'^\star$ if, and only if:
  \begin{itemize}
  \item $\call{C}_h^{\sfree}\subseteq \call{C}_h'^\star$;
  \item and for each ASD
    $\call{C}_\call{I}\in\call{C}_h^{\sfree}$ and for all testing ASD
    $\call{C}_t\in (\call{C}_\call{I}\comp\call{C}_h)^\star$ we have
    $\call{C}_t\in(\call{C}_\call{I}\comp\call{C}_h')^\star$.
  \end{itemize}
\end{lemma} 
\begin{proof}
  Assume
  $(\call{C_I},\psi)\in\call{C}_h^\star\setminus\call{C}_h'^\star$.
  By Lemma~\ref{lem:well:formed} we can assume wlog that
  $\call{C_I}=\defsd{_{\call{I}}}$ is well-formed. By
  Lemma~\ref{equiv:lem:decompose:ASD:connection} \call{C_I} can be
  written $\call{C}_d\comp_\varphi \call{C}_t$ where $\call{C}_d$ is a
  stutter-free ASD and $\call{C}_t$ is a testing ASD. By construction
  we have $(\call{C}_t,\varphi)\in (\call{C}_d\comp_\psi
  \call{C}_h)^\star$.  Since $\call{C}_d \comp_\varphi \call{C}_t =
  \call{C_I} \notin {\call{C}_h'}^\star$ then either
  $\call{C}_d\comp_\psi \call{C}_h'$ is closed, but not satisfiable,
  or $\call{C}_t\comp_\varphi (\call{C}_d \comp_\psi \call{C}_h')$.
  In the former case we have $\call{C}_h^{\sfree}\not\subseteq
  \call{C}_h'^\star$, and in the latter case we have $\call{C}_t\in
  (\call{C}_\call{I}\comp\call{C}_h)^\star \setminus
  (\call{C}_\call{I}\comp\call{C}_h')^\star$.

  Conversely, if one of the two points does not hold, we easily
  construct an ASD in $\call{C}_h^\star\setminus\call{C}_h'^\star$.
\end{proof}

Then we prove that if in the previous lemma the testing part is known,
the stutter-free part is also a stutter-free solution of the
connection between the testing part and the HSD.

\begin{lemma}{\label{equiv:lem:symmetry}}
  Assume $\call{C}_\call{I}\in\call{C}_h^\sfree$ and $\call{C}_t\in
  (\call{C}_\call{I}\comp\call{C}_h)^\star$. Then
  $\call{C}_\call{I}\in (\call{C}_t\comp\call{C}_h)^\sfree$.
\end{lemma}

\begin{proof}
  We let $\call{C}_\call{I}$, $\call{C}_h$, and $\call{C}_t$ be as
  in the statement of the lemma, and denote them as follows:
  $$
  \left\lbrace
    \begin{array}{rcl}
      \call{C}_\call{I}&=&\defsd{_\call{I}}\\
      \call{C}_h&=&\defsd{_h}\\
      \call{C}_t&=&\defsd{_t}\\
    \end{array}
  \right.
  $$
  Since $\call{C}_\call{I}\in\call{C}_h^\sfree$ there exists a
  one-to-one\footnote{Since the connection is closed the mapping is
    total.} mapping
  $\varphi:\Invar_\call{I}\cup\Invar_h\to\Outvar_\call{I}\cup\Outvar_h$
  such that $\call{C}_h'=\call{C}_\call{I}\comp_\varphi\call{C}_h$ is
  closed and satisfiable.  Let us denote $\call{C}_h'=\defsd{_h'}$.
  
  Also by hypothesis there exists a one-to-one mapping $\psi:
  \Invar_h'\cup\Invar_t\to\Outvar_h'\cup\Outvar_t$ such that
  $\call{C}_t\comp_\psi\call{C}_h'$ is closed and satisfiable. Since
  $\call{C}_h'$ is closed the function $\psi$ is actually a mapping
  from $\Invar_t$ to $\Outvar_h'\cup\Outvar_t$. Let $D$ be the subset
  of the domain of $\psi$ of indexes $i$ such that
  $\psi(i)\in\Outvar_\call{I}$, and $\bar{D}$ be its complement in the
  domain of $\psi$.  Let us define from $\psi$ and $D$ two functions:
  $$
  \left\lbrace
    \begin{array}{rcl}
      \psi'&=&\psi_{|\bar{D}}\\
      \varphi'&=&\psi_{|D}\cup\varphi\\
    \end{array}
  \right.
  $$
  Let $\call{C}_h''=\call{C}_h\comp_{\psi'}\call{C}_t$. Since by
  construction
  $$
  \call{C}_\call{I}\comp_{\varphi'} (
  \call{C}_h\comp_{\psi'}\call{C}_t ) = \call{C}_t\comp_\psi
  (\call{C}_h\comp_\varphi\call{C}_\call{I})
  $$
  and $\call{C}_t\in(\call{C}_h\comp_\varphi\call{C}_\call{I})^\star$
  the connection between $\call{C}_\call{I}$ and $\call{C}_h''$ is
  also closed and satisfiable, and thus
  $\call{C}_\call{I}\in(\call{C}_h'')^\star$. Since
  $\call{C}_\call{I}\in\call{C}_h^\sfree$ the first two points of the
  definition of stutter free derivations are satisfied by
  $\call{C}_\call{I}$. Given that:
  $$
  \varphi'_{\Invar_h\cup\Invar_\call{I}} =
  \varphi_{\Invar_h\cup\Invar_\call{I}}
  $$
  it is easy to see that:
  $$
  \trace{\call{C}_\call{I}}{\call{C}_\call{I}\comp_{\varphi'} (
    \call{C}_h\comp_{\psi'}\call{C}_t )}=
  \trace{\call{C}_\call{I}}{\call{C}_\call{I}\comp_{\varphi}
    \call{C}_h}
  $$
  As a consequence the hypothesis
  $\call{C}_\call{I}\in\call{C}_h^\sfree$ implies
  $\call{C}_\call{I}\in(\call{C}_h'')^\sfree$.
\end{proof}

The next step is to bound the size of the testing ASD $\call{C}_t$
obtained in Lemma~\ref{equiv:lem:split}.  To this end, given an ASD
$\call{C}_\call{I}\in\call{C}_h^\sfree$ we define:
$$
\chi(\call{C}_\call{I})=\condset{\call{C}_t \text{ testing ASD}}{\call{C}_t\comp\call{C}_\call{I}\in
  \call{C}_h^\star\setminus {\call{C}_h'}^\star}
$$
\textit{i.e.} the set of testing ASDs that distinguish $\call{C}_h$
from $\call{C}_h'$.  By Lemma~\ref{equiv:lem:split},
$\call{C}_h^\star\not\subseteq{\call{C}_h'}^\star$ if, and only if,
there exists an ASD $\call{C}_\call{I}$ such that
$\chi(\call{C}_\call{I})\neq\emptyset$. By ordering the equations in
the unification system of an ASD
$\call{C}_t\in\chi(\call{C}_{\call{I}})$ and keeping a minimal one, we
prove that an ASD of bounded length can be constructed from
$\call{C}_t$.

\begin{proposition}{\label{equiv:prop:minimal:test}}
  $\call{C}_h^\star \not\subseteq{\call{C}_h'}^\star$ if, and only if, there
  exists $\call{C}_\call{I}\in\call{C}_h^\sfree$ such that $\chi(\call{C}_\call{I})$ contains an
  ASD $\call{C}_t$ with at most one deduction and one equality test.
\end{proposition}
\begin{proof}
  The converse direction is trivial.

  First let us note that if $\call{C}'\in \call{C}_h^\star \setminus
  {\call{C}_h'}^\star$ then, adding test equations to $\call{C}'$
  which are satisfied by \trace{\call{C}'}{\call{C}'\comp\call{C}_h}
  yields another symbolic derivation in $\call{C}'\in \call{C}_h^\star
  \setminus {\call{C}_h'}^\star$. Thus and wlog we let $\call{C}'\in
  \call{C}_h^\star \setminus {\call{C}_h'}^\star$ be an aware ASD.
  According to Lemma~\ref{equiv:lem:decompose:ASD:connection}
  $\call{C}'$ can be split into one stutter-free derivation
  $\call{C}_\call{I}=\defsd{_\call{I}}$ and one test derivation
  $\call{C}_t=\defsd{_t}$.  We also define a partition
  $\call{S}_t^d\cup\call{S}_t^t$ of $\call{S}_t$ such that
  $\call{S}_t^d$ contains only deduction equations and $\call{S}_t^t$
  contains only test equations. Let
  $\call{C}_t^d=(\call{V}_t,\call{S}_t^d,\call{K}_t,\Invar_t,\Outvar_t)$.
  Let us define the following substitutions:
  $$
  \left\lbrace
    \begin{array}{rcl@{\hspace*{3em}}rcl}
      \sigma_\call{I}&=&\trace{\call{C}_\call{I}}{\call{C}_\call{I}\comp\call{C}_h} &
      \sigma_\call{I}'&=&\trace{\call{C}_\call{I}}{\call{C}_\call{I}\comp\call{C}_h'} \\
      \sigma_t&=&\trace{\call{C}_t}{\call{C}_t\comp\call{C}_\call{I}\comp\call{C}_h} &
      \sigma_t'&=&\trace{\call{C}_t'}{\call{C}_t'\comp\call{C}_\call{I}\comp\call{C}_h}\\
    \end{array}
  \right.
  $$
  where the ASD $\call{C}_t'$ is constructed from $\call{C}_t$ as
  follows.  We note that, if $\call{V}_t(i)=\call{V}_t(j)$ for two
  distinct states $i,j$ which are not reuse states, we can introduce a
  new variable $x$, change $\call{V}_t(j)$ to $x$, and introduce in
  $\call{S}_t$ a new test equation $\call{V}_t(i)\unif x$. In other
  words we can assume \textit{wlog} that $\call{V}_t$ is injective on
  states which are not reuse states.  This permits one to ensure that
  the subset $\call{S}_t^d$ of equations which are not test equations
  is satisfiable in any closed connection with another symbolic
  derivation. We define
  $\sigma_t^d=\trace{\call{C}_t^d}{\call{C}_t^d\comp\call{C}_\call{I}\comp\call{C}_h'}$.

  By the second point of
  Lemma~\ref{equiv:lem:decompose:ASD:connection} there exists a
  mapping $\psi: \Ind_t\to\Ind_\call{I}$ such that for every
  $i\in\Ind_t$ we have
  $\call{V}_t(i)\sigma_t=\call{V}_\call{I}(\psi(i))\sigma_\call{I}$.
  \textit{Wlog} we assume that $\psi$ is defined as an extension of
  the connection between $\call{C}_\call{I}$ and $\call{C}_t$, thereby
  ensuring that for input states $i$ of $\call{C}_t$ we also have
  $\call{V}_t(i)\sigma_t'=\call{V}_\call{I}(\psi(i))\sigma_\call{I}'$.

  \begin{claim}
    Wlog we can assume that for any deduction state $i\in\Ind_t$ we
    have
    $\call{V}_t(i)\sigma_t'\neq\call{V}_\call{I}(\psi(i))\sigma_\call{I}'$.
  \end{claim}

  \begin{proofclaim}
    Let $i\in\Ind_t$ be a deduction state such that
    $\call{V}_t(i)\sigma_t' =
    \call{V}_\call{I}(\psi(i))\sigma_\call{I}'$. Adding a reuse state
    if necessary, we can change $i$ into an input state that is
    connected to $\psi(t)$ (or a state which is a reuse of $\psi(i)$).
    This construction does not change $\sigma_t$ nor $\sigma_t'$ and
    thus the fact that
    $\call{C}_t\comp\call{C}_\call{I}\comp\call{C}_h$ or
    $\call{C}_t\comp\call{C}_\call{I}\comp\call{C}_h'$ is satisfiable.
    When repeatedly applying it, we obtain a symbolic derivation
    $\call{C}_t$ that satisfies the claim.
  \end{proofclaim}

  We now split the analysis in two cases depending on whether the set
  $I_t\subseteq\Ind_t$ of indexes $i$ such that
  $\call{V}_t(i)\sigma_t'\neq\call{V}_\call{I}(\psi(i))\sigma_\call{I}'$
  is empty or not.  If it is empty, the claim implies that we can
  assume there is no deduction states in $\call{C}_t$, and thus that
  $\call{S}_t=\call{S}_t^t$. Since
  $\call{C}_t\comp\call{C}_\call{I}\comp\call{C}_h$ is satisfiable but
  not $\call{C}_t\comp\call{C}_\call{I}\comp\call{C}_h'$ there exists
  two input states $i,j$ and one equation
  $\call{V}_t(i)\unif\call{V}_t(j)$ in $\call{S}_t$ which is satisfied
  by $\sigma_t$ but not by $\sigma_t'$. Thus $\chi(\call{C}_\call{I})$
  contains one symbolic derivation $(\call{V}:i\in\set{1,2}\mapsto
  x_i,\set{x_1\unif x_2},\emptyset,\set{1,2},\emptyset)$ where $1$ is
  connected to $\psi(i)$ and $2$ is connected to $\psi(j)$.

  On the other hand, if $I_t$ is not empty, let $i_0$ be minimal in
  this set, and let $\call{V}_t(i_0)\unif
  f(\call{V}_t(i_1),\ldots,\call{V}_t(i_n))$ be the equation
  corresponding to this deduction state in $\call{S}_t^d$.  Given the
  claim we can assume that $i_t$ is the first deduction state, and
  thus that all preceding states are input states. Thus there exists
  an ordering on the set $\Ind_0=\set{t,0,\ldots,n}$ such that the
  following symbolic derivation is in $\chi(\call{C}_\call{I})$ and
  satisfies the proposition:
  $$
  (\call{V}:i\in\Ind_0\mapsto x_i,\set{x_0\unif f(x_1,\ldots,x_n) ~,~
    x_0 \unif x_t},\set{t,1,\ldots,n},\emptyset)
  $$
\end{proof}

Now we simply gather the results from Lemma~\ref{equiv:lem:symmetry}
and Proposition~\ref{equiv:prop:minimal:test}. 

\begin{proposition}{\label{equiv:prop:inclusion}}
  Given two HSDs $\call{C}_h$ and $\call{C}_h'$ we have
  $\call{C}_h^\star\subseteq{\call{C}_h'}^\star$ if, and only if, there exists
  a symbolic testing derivation $\call{C}_t$ with at most one deduction
  state and one equality and a connection $\varphi$ such that
  $(\call{C}_h\comp_\varphi\call{C}_t)^\sfree\subseteq(\call{C}_h'\comp_\varphi\call{C}_t)^\star$.
\end{proposition}

\begin{proof}
  Let us first prove the contrapositive of the direct direction. Let
  $\call{C}_\call{I}$ be an ASD in
  $(\call{C}_h\comp_\varphi\call{C}_t)^\sfree\setminus(\call{C}_h'\comp_\varphi\call{C}_t)^\star$,
  and $\psi$ be a connection such that:
  $$
  \left\lbrace
    \begin{array}{rcl}
      \call{C}_\call{I}\comp_\psi(\call{C}_h\comp_\varphi\call{C}_t) & \hspace*{2em} &
      \text{is closed and satisfiable}\\
      \call{C}_\call{I}\comp_\psi(\call{C}_h'\comp_\varphi\call{C}_t) & \hspace*{2em} &
      \text{is closed and not satisfiable}\\
    \end{array}
  \right.
  $$
  From $\varphi$ and $\psi$ we easily define two connections
  $\varphi'$ and $\psi'$ such that
  $\call{C}_\call{I}\comp_{\varphi'}\call{C}_t$ is an ASD
  $\call{C}_\call{I}'$ such that
  $\call{C}_\call{I}'\comp_{\psi'}\call{C}_h$ is closed and
  satisfiable whereas $\call{C}_\call{I}'\comp_{\psi'}\call{C}_h'$ is
  closed but not satisfiable. Hence:
  $$
  (\call{C}_h\comp_\varphi\call{C}_t)^\sfree\setminus(\call{C}_h'\comp_\varphi\call{C}_t)^\star\neq\emptyset
  $$
  implies $\call{C}_h^\star\not\subseteq {\call{C}_h'}^\star$.

  Let us now prove the contrapositive of the converse implication and
  assume $\call{C}_h^\star\not\subseteq {\call{C}_h'}^\star$. By
  Proposition~\ref{equiv:prop:minimal:test} there exists a symbolic
  derivation $\call{C}_\call{I}\in\call{C}_h^\sfree$, a testing ASD
  $\call{C}_t$ and a connection $\psi$ such that:
  $$
  \left\lbrace
    \begin{array}{l}
      \call{C}_t\comp_\psi\call{C}_\call{I} \in {\call{C}_h}^\star\\
      \call{C}_t\comp_\psi\call{C}_\call{I} \notin {\call{C}_h'}^\star\\
      \call{C}_t \text{ contains at
        most one deduction and one equality test}\\
    \end{array}
  \right.
  $$
  By Lemma~\ref{equiv:lem:symmetry} this implies that there exists a
  connection $\varphi$ such that
  $\call{C}_\call{I}\in(\call{C}_h\comp_\varphi\call{C}_t)^\sfree$.
  Given the construction it is clear that $\call{C}_\call{I}\notin
  (\call{C}_h'\comp_\varphi\call{C}_t)^\star$.
\end{proof}

The proof of the following theorem depends on the fact that for
finitary deduction systems, the set $\min_<(
(\call{C}_t\comp\call{C}_h)^\sfree)$ is by definition finite. The test
of Proposition~\ref{equiv:prop:inclusion} thus becomes effective by
Lemma~\ref{equiv:lem:reduction:sol} when a finite witness set is
available.

\begin{theorem}{(Inclusion of $\call{C}_h^\star$ into ${\call{C}_h'}^\star$)\label{equiv:theo:inclusion}}
  Let \call{D} be a finitary deduction system. The inclusion
  $\call{C}_h^\star\subseteq{\call{C}_h'}^\star$ is decidable for any two honest
  \call D-symbolic derivations $\call{C}_h,\call{C}_h'$.
\end{theorem}

\begin{proof}
  By Prop.~\ref{equiv:prop:inclusion} the inclusion does not hold if,
  and only if, there exists an ASD $\call{C}_t$ of bounded length and
  a connection function $\varphi$ such that:
  $$
  \Delta= (\call{C}_h\comp_\varphi\call{C}_t)^\sfree \setminus
  (\call{C}_h'\comp_\varphi\call{C}_t)^\star \neq\emptyset
  $$
  Let $\call{C}_\tau$ be an ASD in $\Delta$. By definition of finitary
  deduction systems one can compute from
  $\call{C}_h\comp_\varphi\call{C}_t$ a finite set $\Sigma$ of ASDs
  such that there exists $\call{C}_\sigma\in\Sigma$ and $\call{C}_c$
  stutter free such that
  $\call{C}_\call{I}'\le\call{C}_\call{I}\comp\call{C}_c$.  By
  definition of the ordering there exists a stutter free derivation
  $\call{C}_\theta$ and a set of constants $C$ such that:
  $$
  \open{\call{C}_\sigma}{C}\comp\call{C}_\theta=\call{C}_\tau\comp\call{C}_c
  $$
  By hypothesis there exists a connection function $\psi$ such that
  $\call{C}_\tau\comp_\psi(\call{C}_h\comp_\varphi\call{C}_t)$ is
  closed and satisfiable whereas
  $\call{C}_\tau\comp_\psi(\call{C}_h'\comp_\varphi\call{C}_t)$ is
  closed but not satisfiable.  By
  Lemma~\ref{equiv:lem:replace:constant:by:term} (employed with
  $C=\emptyset$)
  $\call{C}_c\comp(\call{C}_\tau\comp_\psi(\call{C}_h\comp_\varphi\call{C}_t))$
  is satisfiable whereas, since
  $\call{C}_\tau\comp_\psi(\call{C}_h'\comp_\varphi\call{C}_t)$ is
  closed,
  $\call{C}_c\comp(\call{C}_\tau\comp_\psi(\call{C}_h'\comp_\varphi\call{C}_t))$
  is not.  By Lemma~\ref{equiv:lem:replace:constant:by:term} if
  $\call{C}_\sigma\in{\call{C}_h'}^\star$ then so is
  $\call{C}_c\comp(\call{C}_\tau\comp_\psi(\call{C}_h'\comp_\varphi\call{C}_t))$.
  Since $\call{C}_\sigma\in\Sigma$ implies
  $\call{C}_\sigma\in(\call{C}_h\comp_\varphi\call{C}_t)^\star$ we
  thus have
  $\call{C}_\sigma\in(\call{C}_h\comp_\varphi\call{C}_t)^\star\setminus
  (\call{C}_h'\comp_\varphi\call{C}_t)^\star$. Thus, if
  $\call{C}_h\not\subseteq\call{C}_h'$ one can guess (in bounded time)
  a symbolic derivation $\call{C}_t$ and compute a finite $\Sigma$ of
  symbolic derivations that contains one which is not in
  $(\call{C}_h'\comp\call{C}_t)^\star$.

  Conversely it is clear if one such derivation is found then
  $\call{C}_h^\star\not\subseteq{\call{C}_h'}^\star$.
\end{proof}

As a trivial consequence we obtain the announced theorem.

\begin{reftheorem}{equiv:theo:dy:decidable}
  Symbolic equivalence is decidable for finitary deduction systems.
\end{reftheorem}

\section{Conclusion}

We have introduced in this paper the notion of \emph{finitary
  deduction systems}, and proved that symbolic equivalence is
decidable for such attacker models. We believe that definition also
captures the essence of \emph{lazy intruder} techniques that are
employed in many tools. Accordingly, we believe that a practical
consequence of this paper will be the inclusion in existing
reachability analysis tools of a symbolic equivalence checking
algorithm.

In terms of comparison of expected runtimes for tools currently
deciding reachability, a back-of-the-enveloppe computation for tools
employing lazy constraint solving techniques such as
OFMC~\cite{DBLP:journals/ijisec/BasinMV05} and CL-AtSe~\cite{T-RTA06}
would be twice (given that two protocols have to be analyzed and
assuming tool is not parallelized) the runtime for safe (since these
tools usually stop at the first attack found, and thus typically have
a much shorter running time in these cases) protocols of a similar
size. We refer the interested reader to~\cite{T-RTA06} for more
details, but given that CL-AtSe now implements a concurrent search
algorithm and has been deployed on Amazon's EC2, we believe that less
than 10s for reasonable industrial protocols is achievable nowadays.

\bibliography{biblio}

\end{document}